\documentclass{article}
\usepackage[utf8]{inputenc}

\usepackage[left=1in,top=1in,right=1in,bottom=1in,nohead]{geometry}

\usepackage{bbm}
\usepackage[overload]{empheq}
\usepackage{amsfonts}
\usepackage{amsmath}
\usepackage{mathtools}
\usepackage{amssymb}
\usepackage{dsfont}
\usepackage{diagbox}
\usepackage{graphicx}
\usepackage{bm}
\usepackage{dsfont, soul, stackrel, tikz-cd, graphicx, titlesec, etoolbox}

\usepackage{tikz, pgfplots}
\usetikzlibrary{automata,positioning,arrows, intersections, calc}
\pgfplotsset{compat=1.8}
\usepackage{caption}
\usepackage{subcaption}
\usepackage{placeins}
\usepackage{latexsym,amsthm,xcolor,graphicx, longtable}
\usepackage{hyperref}
\usepackage{mathtools}
\usepackage{tikz, pgfplots}
\usepackage{makecell}
\usepackage{tikz}
\usetikzlibrary{automata,positioning,arrows, intersections, calc}
\pgfplotsset{compat=1.8}

\newtheorem{thm}{Theorem}

\newtheorem{proposition}{Proposition}

\theoremstyle{definition}
\newtheorem{remark}{Remark}

\def\x{\mathbf x}
\def\xn{\mathbf x_0}

\definecolor{ao(english)}{rgb}{0.0, 0.5, 0.0}

\title{A reaction network model of microscale liquid-liquid phase separation reveals effects of spatial dimension}
\author{Jinyoung Kim\thanks{Department of Mathematics, Pohang University of Science and Technology (POSTECH), Pohang 37673, Republic of Korea} , Sean D. Lawley\thanks{Department of Mathematics, University of Utah, Salt Lake City, UT 84112 USA}, and Jinsu Kim\thanks{Department of Mathematics, Pohang University of Science and Technology (POSTECH), Pohang 37673, Republic of Korea, jinsukim@postech.ac.kr}}
\date{January 2024}

\begin{document}

\maketitle
\tableofcontents

\abstract{
Proteins can form droplets via liquid-liquid phase separation (LLPS) in cells.  Recent experiments demonstrate that LLPS is qualitatively different on two-dimensional (2d) surfaces compared to three-dimensional (3d) solutions. In this paper, we use mathematical modeling to investigate the causes of the discrepancies between LLPS in 2d versus 3d. We model the number of proteins and droplets inducing LLPS by continuous-time Markov chains and use chemical reaction network theory to analyze the model.
To reflect the influence of space dimension, droplet formation and dissociation rates are determined using the first hitting times of diffusing proteins. We first show that our stochastic model reproduces the appropriate phase diagram and is consistent with the relevant thermodynamic constraints. After further analyzing the model, we find that it predicts that the space dimension induces qualitatively different features of LLPS which are consistent with recent experiments. While it has been claimed that the differences between 2d and 3d LLPS stems mainly from different diffusion coefficients, our analysis is independent of the diffusion coefficients of the proteins since we use the stationary model behavior. Therefore, our results give new hypotheses about how space dimension affects LLPS.}

\section{Introduction}

In a cell, LLPS manifests as the formation of droplets from protein condensation. These liquid droplets in a dense phase separate from their surrounding spaces of dilute phases \cite{alberti2018LLPS, banani2017biomolecular, brangwynne2009germline, hyman2014liquid, shin2017liquid}. Biologists have made significant strides in elucidating the importance of LLPS and its involvement in cellular processes. Repair protein factors, for example, are involved in the initiation of LLPS during DNA double-strand breaks \cite{wang2023Break_DNA, levone2021repair, mine2022repair}. Additionally, it has been found that the condensates of SCOTIN, an Endoplasmic Reticulum (ER) transmembrane protein with a cytosolic intrinsically disordered region, inhibit ER-to-Golgi transport through LLPS \cite{kim2023intrinsically}. 
Furthermore, dysregulation of LLPS has been associated with various diseases, including cancer \cite{Mehta2022Disease, Wang2021Disease, Wegmann2018Tau}.

LLPS is found in many different intracellular locations, including on 2-dimensional (2d) surfaces such as the endoplasmic reticulum membrane  \cite{ snead2019control, snead2022membrane, kim2023intrinsically} or in 3-dimensional (3d) spaces such as the cytoplasm \cite{kedersha2013stress, snead2019control, Lin2015RNA, hirose2023guide} (Figure \ref{fig:overall}A). Qualitative differences between 2-dimensional LLPS (2d-LLPS) and 3-dimensional LLPS (3d-LLPS) have been observed recently 
\cite{Jonathon2021threshold, snead2022membrane, case2022membranes}. For example, Snead et al.\ \cite{snead2022membrane} revealed differences in droplet formation times between 2d and 3d environments. While droplets in 2d can form within minutes, it takes hours for them to form in 3d. Additionally, droplets in 2d can be arrested within minutes, suggesting resistance to size growth, whereas in 3d, they can reach their maximum size within hours.
To explain these qualitative differences observed in 2d versus 3d, some have proposed that droplet size arrest may result from disparities of diffusion coefficients in 2d and 3d environments \cite{snead2022membrane}. For example, diffusion coefficients of eGFP proteins in 2d differ markedly from 3d  \cite{thomas2015Diff_2D, Zareh2012Diff3D_2D}.

While it is true that diffusion coefficients differ in 2d versus 3d cellular environments, the space dimension has more fundamental effects on diffusion processes. For example, far fewer steps are required (on average) for a random walk to find a target if the search is restricted to a 2d lattice rather than a 3d lattice \cite{adam68}. More mathematically sophisticated examples of dimensional discrepancies include the recurrence versus transience of 2d versus 3d random walks \cite{NorrisMC97, Durrett_process} and the logarithmic versus algebraic singularities of the 2d versus 3d Laplacian Green's functions \cite{evans2022partial}. Importantly, these discrepancies cannot be accounted for by merely rescaling time. Studying how these fundamental, dimensional differences in diffusion affect cell biology has a long history in the biophysics literature \cite{saxton2007modeling,nadler1991,abel12,shea97,kholodenko00,haugh02,monine08,alonso10,lawley2016pi,lawley2017pi,dixon2024,grebenkov2022search}. How do these differences affect LLPS?

In this paper, we formulate and analyze mathematical models of microscopic intracellular LLPS and theoretically compare LLPS in 2d versus 3d. The main theoretical framework for this is biochemical reaction network theory. By describing droplet formation, coarsening, and dissociation with reactions between species, LLPS can be associated with a reaction network. Our model consists of the following reactions, 
\begin{align}\label{eq:main original crn}
\begin{split}
   &m P\xrightleftharpoons[b_m]{a_1} D_m, \quad \text{(droplet formation with $m$ copies of proteins and dilution)}\\
   &P+D_k\xrightleftharpoons[b_{k+1}]{a_k} D_{k+1}, \quad \text{for $k=m,\dots,L-1$ (coarsening and dissociation)}\\
   &D_{k}+D_{j}\xrightleftharpoons[g_{k+j}]{f_{kj}} D_{k+j}, \quad \text{for $2m \le k+j\le L$ (fusion and fission)}
   \end{split}
\end{align}
where $P$ denotes a single protein and $D_k$ denotes a droplet consisting of $k$ proteins.  $m$ indicates the threshold number of proteins to form a droplet. We model the reactions in \eqref{eq:main original crn} with a continuous-time Markov chain which tracks the copy numbers of each ``species'' $P$ and $D_i$. 
We derive the closed form of the stationary probability distribution of this Markov chain using biochemical reaction network theory. 
The reaction rate parameters ($a_i, b_i, f_i$, and $g_i$ in \eqref{eq:main crn}) are set via first passage time theory of diffusion processes to reflect spatial dimension differences. We then study how these dimensional disparities in diffusion yield differences for LLPS in relation to the reaction rates by computing the resulting stationary distribution of \eqref{eq:main original crn}. 

The theoretical study of LLPS spans various fields. In physical chemistry, scientists have investigated LLPS phenomena under thermodynamic theory by measuring energy, showing that energy minimization leads to the demixing of substances and liquid state phase separation \cite{bartolucci2024interplay, hyman2014liquid, brangwynne2015polymer, falahati2019thermodynamically}. Additionally, theorists used partial differential equation models, such as the Cahn-Hilliard equation \cite{cahn1958cahnhilliard}, Allen-Cahn equation \cite{allen1979phaseboundary} and Cahn-Hilliard-Navier-Stokes equation \cite{jacqmin1999CahnNavier}, to analyze and numerically simulate LLPS \cite{zhang2008Energy, gasior2020modeling, gasior2019partial, hester2023Energy, yu2021hsp70, girelli2021microscopic, hester2023fluid}. Machine learning and data-driven methods are also employed to analyze phase separation \cite{saar2023theoretical}.
In contrast to previous models that primarily use thermodynamic frameworks like free energy and chemical potentials to explain LLPS, our model is built from first hitting times of diffusing proteins.

We now briefly summarize our results and their biophysical implications. We first verify that our model reproduces the appropriate phase diagram and phase separation and is consistent with existing thermodynamic models. We then study qualitative differences between 2d-LLPS and 3d-LLPS via the stationary distribution of the reaction network \eqref{eq:main original crn}. The shape of the stationary distribution is determined by protein characteristics such as the droplet viscosity, the minimum size of droplets, and the hydrodynamic radius of proteins. We first investigate the effect of the droplet viscosity, indicating the strength of the protein-protein interactions. Within a wide range of the droplet viscosity, 2d-LLPS forms large droplets while proteins are likely to remain without forming in 3d-LLPS. Notably, a higher droplet viscosity is required in 3d than in 2d to increase droplet size. Next, we found that there exists a range of the minimum droplet size in which 2d and 3d LLPS have significantly different probabilities of forming droplets. Finally, we showed that when proteins are tethered on a membrane yielding a reduction on the hydrodynamic radius of the protein, less droplets in 2d can be produced compared to a 3d space, but only if the reduction is significant enough. We display these results using both mathematical analysis and numerical computations.

The stationary distribution thus reveals how diffusion in 2d versus 3d yields differences between 2d-LLPS and 3d-LLPS. Importantly, the stationary distribution is independent of the diffusion coefficient. Therefore, our analysis predicts that prominent qualitative differences between 2d-LLPS and 3d-LLPS stem from fundamental differences in spatial dimension rather than solely from differences in diffusion coefficients. To our knowledge, our study provides the first model of intracellular LLPS using first passage time analysis, chemical reaction network theory, and continuous-time Markov chains.

This manuscript is outlined as follows. We first introduce biochemical reaction networks, one of the key theoretical frameworks of this study, in Section \ref{sec:crn}. In that section, we also derive the closed form stationary distribution of the copy numbers of the proteins and the droplets. In Section \ref{sec:rates}, we use first passage time theory for setting the reaction rates. In Section \ref{sec:main results}, the main results are provided: reproduction of thermodynamic description of LLPS with our model and  the qualitative differences of stationary distributions modeling 2d and 3d LLPS in terms of the viscosity, the threshold droplet size, and the hydrodynamic radius of proteins. In Section \ref{sec:analysis}, we provide mathematical analyses of our main results.

\begin{figure}[htbp]
    \centering
\includegraphics[width=0.8\textwidth]{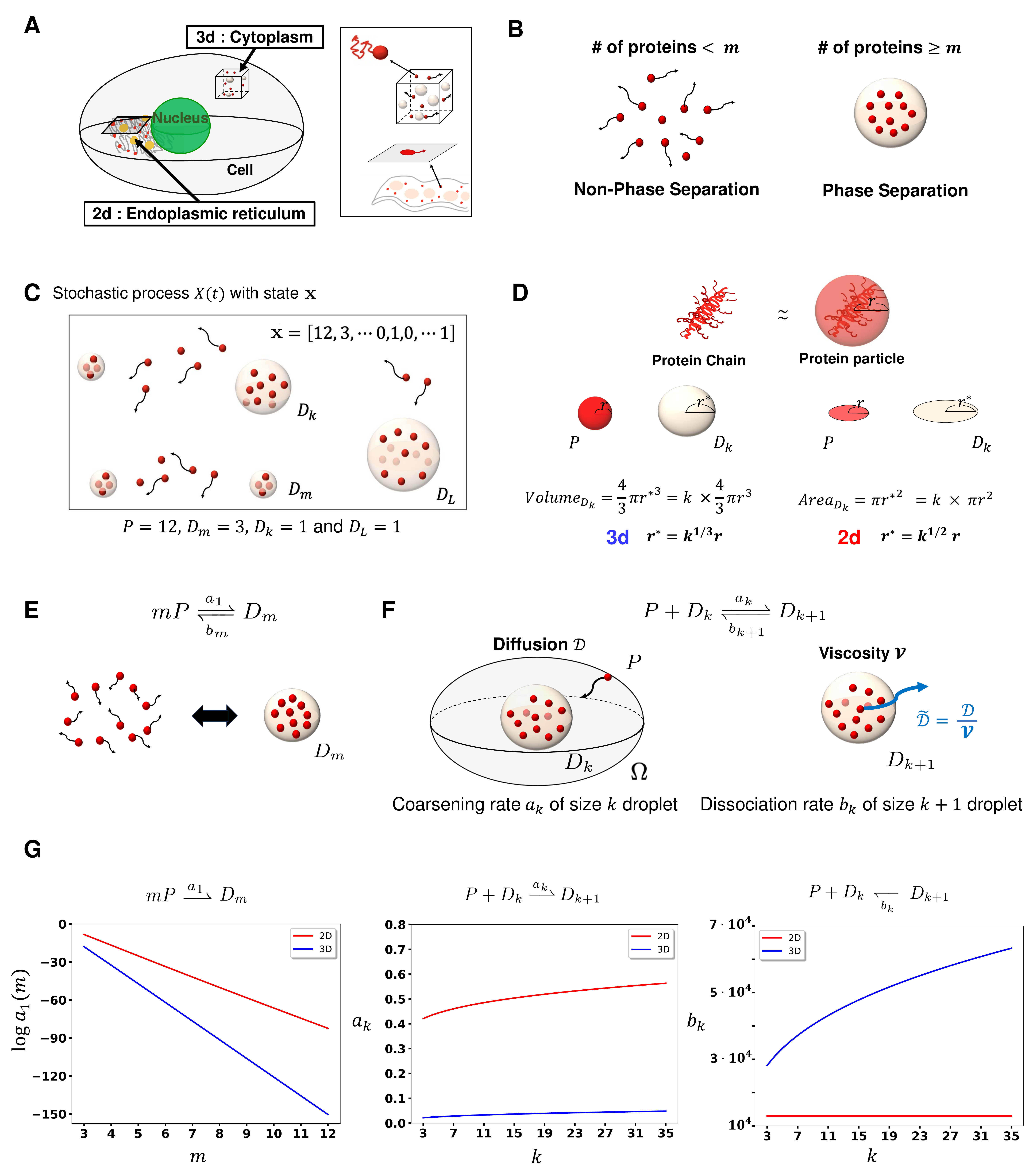} 
\caption{ \textbf{Overview of modeling LLPS using a reaction network.} 
\textbf{A.} LLPS occurrence in 2d (e.g., membrane) and 3d (e.g., cytoplasm). \textbf{B.} The value $m$ represents the threshold number required for droplet formation. \textbf{C.} A visual representation of a state of the Markov chain $X(t)$. \textbf{D.} The hydrodynamic radius $r$ of a protein. A droplet is approximated as a sphere in 3d and a circle in 2d to derive its radius. \textbf{E-F.} Understanding reaction the rates $a_1, a_k$'s and $b_k$'s. The $a_k$ and $b_k$ are defined by mean first hitting and exit times. \textbf{G.} Plots of the rates $a_1, a_k, b_k$ of 2d and 3d under equal diffusion coefficients and viscosity.}
\label{fig:overall}
\end{figure}


\section{Reaction network description of LLPS}\label{sec:crn}

We develop a stochastic process modeling LLPS based on reaction networks to describe LLPS in both 2d and 3d cellular environments. A reaction network is a graph whose nodes and edges represent complexes and reactions, respectively. For example, in \eqref{eq:main crn}, the reaction $P + D_k\to D_{k+1}$ describes coarsening of the droplet of $k$ proteins by recruiting an additional protein $P$. The reactant $P+D_k$ is a complex consisting of a single copy of $P$ and a single copy of $D_k$, and $D_{k+1}$ is the product complex of the reaction.

We use a continuous-time Markov chain to model the stochastic evolution of the copy numbers of species in a reaction network. Specifically, let $X(t)=(P(t),D_m(t),\dots,D_L(t))$ be a continuous-time Markov chain associated with \eqref{eq:main original crn}. Each coordinate of $X(t)$ gives the copy number of the corresponding species at time $t$ (Figure \ref{fig:overall}C). The evolution of $X$ is given by a reaction. For example, if $P+D_m\to D_{m+1}$ fires at $t$, then $X(t)-X(t-)=(-1,-1,1,0,\dots,0)$. The reaction to fire and the time for the next reaction are randomly determined using the reaction intensity $\lambda_{y\to y'}$ for a reaction $y\to y'$ defined as
\begin{align}\label{eq:transition rates }
    P(X(t+\Delta t)=x+\nu_{y\to y'}| X(t)=x) = \lambda_{y\to y'}(x) \Delta t+ o(\Delta t)\quad\text{as }\Delta t\to0+,
\end{align}
where $\nu_{y\to y'}$ is the \emph{reaction vector}  describing the net change of the  reaction $y\to y'$. For example, $\nu_{mP\to D_m}=(-m,1,0,0,\dots,0)$. The function $\lambda_{y\to y'}$ provides the rate of the transition given by the reaction $y\to y'$ \cite{NorrisMC97}. Hence these intensities fully characterize $X$. 
We highlight important assumptions for modeling LLPS.
\begin{itemize}
	\item[(1)] We assume that the timescale of LLPS is faster than protein production and degradation. Hence, we do not consider production and degradation reactions, $P\rightleftharpoons \emptyset$.
	\item[(2)] We do not consider reactions $kP\rightleftharpoons D_k$ for $k<m$ because we assume that there exists a threshold number, $m$, of proteins to form a droplet (Figure \ref{fig:overall}B). The existence of such a threshold was experimentally and theoretically verified in \cite{martin2021multi}.  The volume fraction of the dense phase under the stationary distribution of $X(t)$, which we drive in \eqref{eq:pi}, can be also used to theoretically support this. We provide more details about this setting in Appendix \ref{sub_app:threshold}. 
     \item[(3)] Proteins inside droplets are expected to have smaller mobility compared to proteins outside droplets \cite{kamagata2022guest, kim2023intrinsically}, which means that $f_{kj}$ and $g_{k+j}$ are much smaller than the other reaction rates. Hence, to simplify our analysis, we neglect the fusion reactions and the fission reactions $D_{k}+D_{j}\xrightleftharpoons[g_{k+j}]{f_{kj}} D_{k+j}$ (see Appendix \ref{sub_app:mobility} for details on this assumption).

\end{itemize}
Under these assumptions, the reaction network describing LLPS in this paper is
\begin{align}\label{eq:main crn}
\begin{split}
    m P \xrightleftharpoons[b_m]{a_1} D_m, \quad \quad 
    P+D_k \xrightleftharpoons[b_{k+1}]{a_k} D_{k+1} \ \text{for $k=m,\dots,L-1$}.
\end{split} 
\end{align}

\begin{remark}
The reactions in \eqref{eq:main crn} do not result in chemical changes on the proteins or the droplets, 
although such reactions are often termed as `chemical reactions' in  mathematical biology or chemical reaction network theory. 
For example, the birth of an animal can be described with a chemical reaction $A\to 2A$. In the same sense, we also emphasize that the reactions in \eqref{eq:main crn} do not mean protein assembly, which is a chemical process related to, but distinct from LLPS. We discuss protein assembly in relation to our modeling in Appendix \ref{app:assembly}.
\end{remark}

In Section \ref{sec:reproduction} and Appendix \ref{app:model details}, we describe how this model is consistent 
with certain thermodynamic aspects of LLPS. 
Based on mass-action kinetics, the intensities of the reactions in \eqref{eq:main crn} are defined at $\x=(p,d_m,\dots,d_L)$ as
\begin{align*}
&\lambda_{mP\to D_m}(\x)=a_1 p(p-1)\cdots (p-m+1)\mathbbm{1}_{p\ge m}, \quad \quad &\lambda_{D_m\to mP}(\x)=b_m d_m \mathbbm{1}_{d_m\ge 1},\\
&\lambda_{P+D_k\to D_{k+1}}(\x)=a_k p d_k \mathbbm{1}_{p\ge 1}\mathbbm{1}_{d_k\ge 1}, \text{ and }  &\lambda_{D_{k+1}\to P+D_k}(\x)=b_{k+1} d_{k+1}\mathbbm{1}_{d_{k+1}\ge 1},
\end{align*} for $k=m,\dots,L-1$.

Abusing notation, for a reaction $y\to y'$, we can regard the complexes $y$ and $y'$ as vectors. For example, the complexes $mP$ and $P+D_m$ can be represented by $(m,0,\dots,0)$ and $(1,1,0,\dots,0)$, respectively. Then for $y\to y'$,  the reaction vector can be denoted by $y'-y$, which means the net gain of species via the reaction $y\to y'$.
The probability distribution $p(\x,t)=P(X(t)=\x)$ of $X(t)$ is governed by the \emph{chemical master equation}, a  system of ordinary differential equations defined as \eqref{eq:master_eq}
\begin{align}\label{eq:master_eq}
    \frac{d}{dt}p(\x,t)&=
    \lambda_{mP\to D_m}(\x-(-m,1,0,\dots,0))p(\x-(-m,1,0,\dots,0),t)\\
    &+\lambda_{D_m\to mP}(\x-(m,-1,0,\dots,0))p(\x-(m,-1,0,\dots,0),t) \notag\\
    &+
    \sum_{k=m}^{L-1} \left( \lambda_{P+D_k\to D_{k+1}}(\x-\nu^+_k)p(\x-\nu^+_k,t) +\lambda_{D_{k+1}\to P+D_k}(\x-\nu^-_k)p(\x-\nu^-_k,t)\right )\notag\\
    &-\left(\lambda_{mP\to D_m}(\x)+\lambda_{D_m\to mP}(\x)  +\sum_{k=m}^{L-1} \left( \lambda_{P+D_k\to D_{k+1}}(\x) +\lambda_{D_{k+1}\to P+D_k}(\x)\right )  \right ) p(\x,t), \notag
\end{align}

where $\nu^+_k$ and $\nu^-_k$ are the reaction vectors associated with $P+D_k\to D_{k+1}$ and $D_{k+1}\to P+D_k$, respectively for each $k$.
\begin{remark}
    A well-known reaction network, the so-called Becker-D\" oring model, has a similar reaction network structure to \eqref{eq:main crn} \cite{ball1988Beck, hingant2016Beck, hingant2021Beck_QSD}. This model is often employed to describe particle aggregation. However, due to the absence of the threshold of the protein concentration, this prior model is limited to describing protein assembly rather than phase separation (see \cite{banani2017biomolecular} for the difference between protein assembly and phase separation).
\end{remark}

\subsection{Stationary distributions}
We analyze the differences between 2d-LLPS and 3d-LLPS using their stationary distributions. The stationary distribution $\pi$ is the limiting distribution of $p(\x,t)$ defined as 
\begin{align*}
    \lim_{t\to \infty}p(\x,t)=\pi(\x) \quad \text{for each $\x$}.
\end{align*}

One advantage of the chemical reaction network description of LLPS is that we can obtain the closed form of $\pi$. To do this, we use Theorem \ref{thm:def 0} \cite{AndProdForm} (see Appendix \ref{app:def 0}). 

Using Theorem \ref{thm:def 0}, the stationary distribution of the associated Markov chain for \eqref{eq:main crn} is for each $ \x =(p,d_m,\dots,d_L)$
\begin{align}\label{eq:pi}
\pi(\x)
&= M \dfrac{1}{p!}\prod_{k=m}^L	\dfrac{Q_k^{d_k}}{d_k!}
\end{align} for each state $\mathbf{x}$, where 
$ Q_k$ is 
\begin{align}\label{eq:Q_k}
	Q_k : =  \begin{cases}
\left (\dfrac{a_1}{b_m} \right )  & \text{if $k =m$},\\
\bigg(\dfrac{a_1a_m, \cdots a_{k-1}}{b_mb_{m+1} \cdots b_k} \bigg) & \text{if $m < k\le L.$}\\
	\end{cases}
\end{align} 
The constant $M$ is the normalizing constant such that
\begin{align*}
    M=\left (\sum_{\mathbf{x} \in \mathbb S_{\mathbf{x}}} \dfrac{1}{p!}\prod_{k=m }^L	\dfrac{Q_k^{d_k}}{d_k!} \right )^{-1},
\end{align*}
where $\mathbb S_{\mathbf{x}}$ is the state space containing $\mathbf{x}$.


One of the advantages of the closed form of the stationary distribution \eqref{eq:pi} for our model is its ability to generalize the numerical results shown in Section \ref{sec:main results}.
    Indeed, the threshold number of proteins to form droplets ($m$) and the size of the largest droplets ($L$) can vary widely depending on several factors including the concentration of the proteins, their affinities to each other, and the specific conditions of the system \cite{yewdall2021coacervates, qin2017protein}.  
    Owing to computational costs, we often use small qualities of $m$, and $L$ for simulations in this paper since the size and complexity of the state space of the model grows rapidly with these parameters. However, due to the closed form of $\pi$ in \eqref{eq:pi}, we can show that our main results hold for general $m$ and $L$ (see Propositions \ref{prop:m}--\ref{prop:r}).


\section{Reaction rates}\label{sec:rates}

We choose the reaction rates in our model (\ref{eq:main crn}) by regarding each protein as a randomly diffusing particle. We consider either a disc (for 2d-LLPS) or a sphere (for 3d-LLPS) surrounding droplet to present a target of particles (proteins) as shown in Figure \ref{fig:overall}D. We view a protein as a circular/spherical object with the \emph{hydrodynamic radius} $r$ that takes into account the hydrodynamic length of the protein or the interaction range of a single protein as shown in Figure \ref{fig:overall}D. Assuming the volume of a droplet $D_k$ is proportional to the number of proteins $k$,  the radius $r_{k,d}$ of the droplet $D_k$ in $d$-dimensional space satisfies
   \begin{align}\label{eq:radi}
        r_{k,d} =\begin{cases}
            \alpha r k^{1/2} \quad \text{if $d=2$},\\
           \alpha r k^{1/3} \quad \text{if $d=3$},
        \end{cases}
    \end{align}
    for each $k=m, m+1, \dots, L$. (See Figure \ref{fig:overall}D for this derivation).
    The proportionality constant $\alpha$ may vary by the binding affinity and density of the protein aggregation. We simply set $\alpha=1$ throughout this manuscript. 

We assume that the reaction $P+D_k \to D_{k+1}$ fires when a particle (protein) hits the target $D_k$. Hence, the rates $a_k$ can be defined using first hitting times. Similarly, we define $a_1$ and $b_k$'s regarding the proteins as diffusing particles.

Throughout this manuscript, $\Vert v \Vert=\sqrt{\sum_{i=1}^K v_i^2}$ denotes the standard Euclidean norm of a vector $v \in \mathbb R^K$. Furthermore, $d=2$ or 2d and $d=3$ or 3d indicate 2-dimensional LLPS and 3-dimensional LLPS respectively.

\subsection{The initial droplet formation rate, \texorpdfstring{$a_1$}.}\label{sec:a1}
The main idea for the generalized Smoluchowski framework introduced in \cite{flegg2016smoluchowski} is to consider $m$ independent $d$-dimensional Brownian particles $B_i(t)$ ($i=1,2\dots,m$) with diffusion coefficient $\mathcal{D}$ within the spherical or circular system domain $\Omega=\{x: \Vert x \Vert \le R\}$ for some $R>0$. In \cite{flegg2016smoluchowski}, probability fluxes were used to determine the rate constant for $m$ particles to be in proximity. It was also shown that a Markov chain under mass-action kinetics with the generalized Smoluchowski rate can closely approximate the same system modeled with Brownian particles \cite{flegg2016smoluchowski}. Note that in general the shape of the domain is irrelevant when the particle is sufficiently small relative to the domain. Then the reaction rate $a_1$ can be set as 
\begin{align*}  a_1=\bigg[ \mathcal{T}_m  \dfrac{4 \pi^{\alpha_d +1}r_{m,d}^{2\alpha_d} }{\Gamma(\alpha_d)}  \bigg] / V^{m-1},
\end{align*}  
where $r_{m,d}$ is as \eqref{eq:radi}, 
 $\mathcal{T}_m = \mathcal{D}  \times \frac{m^{3/2}}{m!}\bigg(\dfrac{m-1}{2} \bigg)^{(3m-5)/2}$, $\alpha_2 = (m-2)$, $\alpha_3 = (3m-5)/2$, $\Gamma$ is the gamma function, and $V$ is the volume of the system domain (see \cite[Equation (3.22)]{flegg2016smoluchowski}). Note that
\begin{align*}
    V = \begin{cases}
         \pi R^2  & \quad \text{if $d=2$},\\
        \dfrac{4}{3}\pi R^3  &  \text{if $d=3$},
    \end{cases}
\end{align*} 
Then \eqref{eq:radi} yields that

\begin{align}\label{eq:a_1}
	a_1 =\begin{cases} \mathcal{T}_m \bigg[\dfrac{4m^{m-2} }{\Gamma(m-2)}  \bigg] \bigg( \dfrac{r}{R} \bigg)^{2m-2} \dfrac{1}{r^{2}} & \text{if $d=2$},\\
 \mathcal{T}_m\bigg[ \bigg(\dfrac{3}{4} \sqrt{\pi}\bigg)^{m-1} \dfrac{4 m^{(3m-5)/3} }{\Gamma((3m-5)/2)} \bigg]
	\bigg( \dfrac{r}{R} \bigg)^{3m-3} \dfrac{1}{r^{2}}  &\text{if $d=3$},
	\end{cases}
\end{align}
 
We can now find that for given any $m$, the rate of the initial droplet formation $a_1$ in 2d-LLPS is much greater than that of 3d-LLPS if proteins have the same diffusion coefficient  $r \ll R$, i.e., the generating time of a droplet in 2d is faster than in 3d  (Figure \ref{fig:overall}G). Note that all occurrences of $m$ are greater than 3 throughout this paper.

\begin{remark}
    The term $m!$ in $\mathcal T_m$ comes from mass-action kinetic for $mP\to D_m$. The intensity of $mP\to D_m$ under mass-action kinetics is combinatorially defined as it is proportional to
    $\begin{pmatrix}
        P \\ m
    \end{pmatrix}=\dfrac{P(P-1)\dots (P-m+1)}{m!}$.
    Hence we merge the term $1/m!$ to $a_1$. 
\end{remark}

\subsection{The droplet coarsening rate, \texorpdfstring{$a_k$} 
 \ \ for \texorpdfstring{$m+1\le k \le L-1$}.}
In our LLPS models, droplet coarsening happens when a protein hits the droplet. We thus model the rates of the coarsening reaction $P+D_k\xrightarrow{a_k} D_{k+1}$ using the mean first hitting time for a protein to hit a droplet $D_k$ (Figure \ref{fig:overall}F). Hence we first consider the $d$-dimensional annular domain, \begin{align*}
	\Omega_{k,d} := \{ x \in \mathbb{R}^d :  r_{k,d} \leq \lVert x \lVert \leq R \}.
\end{align*} 
As Section \ref{sec:a1}, an individual protein is described by a Brownian motion $B(t)$ with the diffusion coefficient $\mathcal D$. Let $\tau_{k,d}$ denote the first hitting time in a $d$-dimensional space, and then it is defined as  \begin{align}\label{eq:ftp}
	\tau_{k,d} = \inf \{ t \geq 0 :  \lVert B(t) \lVert = r_{k,d}\}.
\end{align}
For $k\in \{m,m+1,\dots,L-1\}$, we have that \cite{Sean2020Rate_a_k},
\begin{align}\label{eq:ak}
 \dfrac{1}{\mathbb E[\tau_{k,d}]} = a_{k} = \begin{cases}
 \dfrac{\mathcal (2\mathcal D/R^2)}{\ln {(r/R)}^{-1} + \ln k^{-1/2}} & \text{if $d=2$},\\
		\\
\dfrac{3\mathcal D}{R^2} k^{1/3} \dfrac{r}{R} &  \text{if $d=3$},
	\end{cases}
\end{align}
The rate of coarsening can also be obtained as a Smoluchowski reaction rate \cite{smoluchowski1918versuch}. It is shown that the Smoluchowski reaction rate for two particles is proportional to the inverse of the mean first hitting time $\mathbb E[\tau_{k,d}]$, implying  consistency between the $a_k$'s and the Smoluchowski rate \cite{shoup_szabo1982Smolu_Hitting}, \cite{Paul2014RateA_k}.
\begin{remark}
  In  the formulation of $\tau_{k,d}$ in \eqref{eq:ftp}, the circular (spherical) target is assumed to be centered at the origin. However by the Markovity of $B(t)$, the initial positions of $B(t)$ and the target $D_k$ are negligible for the derivation of \eqref{eq:ak} provided the target is sufficiently small relative to the size of the domain $\Omega$. In this vein, we need not assume that the domain $\Omega$ is either circular or spherical.
\end{remark}

\subsection{The dissociation rate, \texorpdfstring{$b_k$}.}\label{sec:bi}
We now set the rate of the dissociation reactions $D_{k+1}\xrightarrow{b_k} P+D_k$. One of the key features in LLPS is that the droplets are liquid. Hence we determine $b_k$ using the first time when a diffusing protein inside a droplet hits the boundary of the droplet (Figure \ref{fig:overall}F). 

To define this rate, we introduce important factors of LLPS:  droplet viscosity, surface tension, protein valency, and protein binding affinity. The droplet viscosity is determined by how tightly the proteins are bound to each other inside droplets, which in turn can be determined by valency and binding affinity \cite{polyansky2023protein}.  
The valency of a protein is the number of components of the protein that can bind to others. Binding affinity indicates the strength of the bonds.

Intuitively, proteins forming a droplet will have much smaller mobility inside the droplet. If the droplet is made with high valency and high binding affinity proteins, the viscosity inside the droplet is high, leading to even lower mobility of the proteins inside the droplet. Additionally, the principles of energy minimization related to surface energy, which are intrinsically linked to surface tension \cite{marchand2011surface_tension}, imply that a droplet tends to adopt a circular shape in 2d or a spherical shape in 3d \cite{brown1947surface_shape}. Consequently, it is difficult for proteins inside the droplet to disrupt these shapes. Therefore, we simply incorporate the effects of viscosity, valency, the binding affinity of proteins and surface tension by setting the diffusion coefficient of the protein to be inversely proportional to a dimensionless constant $\mathcal V$ \cite{polyansky2023protein}, which we call the viscosity constant. 
 In particular, we set the diffusion coefficient of the protein in a droplet as $\tilde{\mathcal D} = \dfrac{\mathcal D}{\mathcal{V}}$.
 
We now consider the first time for a Brownian particle to exit a disc (or sphere) of radius $r_{k,d}$.
Let a Brownian particle be initiated at an arbitrary location inside the disc or sphere. 
Then we will consider the first time when the particle hits the boundary of the disc or the sphere of radius $r_{k,d}$.
For each $k\in \{m,m+1,\dots,L\}$, let 
\begin{align*}
\tilde{\tau}_{k, d}= \inf \{ t \geq 0 : \Vert B(t) \lVert \geq  r_{k,d} \},  
\end{align*}
 where $B(t)$ is a Brownian motion in $\mathbb R^d$ with the diffusion coefficient $\tilde{\mathcal{D}}$, and $r_{k,d}$ is defined as in \eqref{eq:radi}. 
The mean first passage time $\tilde \tau_{k,d}$ is \cite{Carr2020diffusion}
\begin{equation}
		\dfrac{1}{\mathbb{E}[ \tilde{\tau}_{k, d} ] } =\begin{cases}
		\dfrac{4\tilde{\mathcal D} }{(r_{k,d})^2} & \quad \text{if $d=2$},\\
		\\
		\dfrac{6\tilde{\mathcal D} }{(r_{k,d})^2} & \quad \text{if $d=3$},
	\end{cases}
\end{equation} Since each $D_k$ contains $k$ proteins, we multiply $k$ to the mean first passage time to set $b_k$. Then
\begin{align*}
    \frac{k}{\mathbb{E} [ \tilde{\tau}_{k, d}]} = b_k  =\begin{cases}
		\dfrac{4k\tilde{\mathcal D} }{(r_{k,d})^2}=\dfrac{4 \mathcal D}{\mathcal V r^2} & \quad \text{if $d=2$},\\
		\\
		\dfrac{6k\tilde{\mathcal D} }{(r_{k,d})^2} =\dfrac{6 \mathcal D k^{1/3}}{\mathcal V r^2}& \quad \text{if $d=3$}.
	\end{cases}
\end{align*}
 Note that the dissociation rate constant $b_k$ does not depend on $k$ in  2d while it does in 3d.

\subsection{Analysis of rate constants}\label{sec:intuition}

To summarize, the reaction rates are given by
\begin{align}\label{eq:ratesummary}
        a_1/t_0 
        &=\begin{cases} \gamma_{2} \varepsilon^{2\alpha_2} \quad & \text{in 2d},\\
 \gamma_{3} 
	\varepsilon^{2\alpha_3}  &\text{in 3d},
	\end{cases}\\
 a_{k}/t_0 
 &= \begin{cases}
 -2/\ln \varepsilon & \text{in 2d for $k\ge m+1$},\\
		\\
3 \varepsilon k^{1/3} & \text{in 3d for $k\ge m+1$},
	\end{cases}\\
 b_k/t_0  
 &=\begin{cases}
		4\dfrac{1}{\mathcal{V}}\varepsilon^{-2} & \text{in 2d},\\
		\\
		6\dfrac{1}{\mathcal{V}}\varepsilon^{-2} k^{1/3}& \text{in 3d},
	\end{cases}
    \end{align}
where $t_0=\mathcal{D}/R^2$ denotes the diffusion timescale, $\varepsilon=r/R\ll1$ measures the lengthscale of protein interactions to the size of the confining spatial domain, $\mathcal{V}\gg1$ measures how protein diffusion slows in droplets, $\alpha_2=m-2$, $\alpha_3=(3m-5)/2$, and
\begin{align*}
    \gamma_{d}
    =\begin{cases}
        m^{3/2}\bigg(\dfrac{m-1}{2} \bigg)^{(3m-5)/2} \bigg[\dfrac{4m^{m-2} }{\Gamma(m-2)}  \bigg] & \text{if }d=2,\\
        m^{3/2}\bigg(\dfrac{m-1}{2} \bigg)^{(3m-5)/2}\bigg[ \bigg(\dfrac{3}{4} \sqrt{\pi}\bigg)^{m-1} \dfrac{4 m^{(3m-5)/3} }{\Gamma((3m-5)/2)} \bigg] & \text{if }d=3.
    \end{cases}
\end{align*}
Note that we have ignored the higher order $k$ dependence in $a_k$ in 2d since we assume that $\varepsilon\ll1$.

There are several things to notice from \eqref{eq:ratesummary}. First, 
\begin{align*}
    a_1^{\text{3d}}/a_1^{\text{2d}}=O(\varepsilon^{m-1}), \quad\text{as }\varepsilon\to0.
\end{align*}
Hence, the formation rate of an initial droplet consisting of $m$ proteins is much faster in 2d than 3d. Second, for $k\ge m+1$,
\begin{align*}a_k^{\text{3d}}/a_k^{\text{2d}}=O(\varepsilon\ln\varepsilon)\quad\text{as }\varepsilon\to 0.
\end{align*}
Hence, droplet coarsening is also faster in 2d than 3d, though the difference between 2d and 3d is not as pronounced as it is for the initial droplet formation rates $a_1^{\text{2d}}$ and $a_1^{\text{3d}}$. Third, 
\begin{align*}
b_k^{\text{3d}}/b_k^{\text{2d}}=O(1)\quad\text{as }\varepsilon\to0.
\end{align*}
Hence, droplet dissociation in 2d and 3d occur at similar rates. Finally, the rate of droplet coarsening and dissociation grow with droplet size $k$ in 3d, but these rates are independent of the droplet size $k$ in 2d (to leading order for $\varepsilon\ll1$).
From this analysis, we can expect that 2d is more favorable for phase separation than 3d as schematically described in Figure \ref{fig:scheme}A. The main results of this paper, which are given in Section \ref{sec:main results}, depend on this analysis (see Figure \ref{fig:scheme}B for a schematic summary).

\begin{figure}[ht]
    \centering
\includegraphics[width=0.93\textwidth]{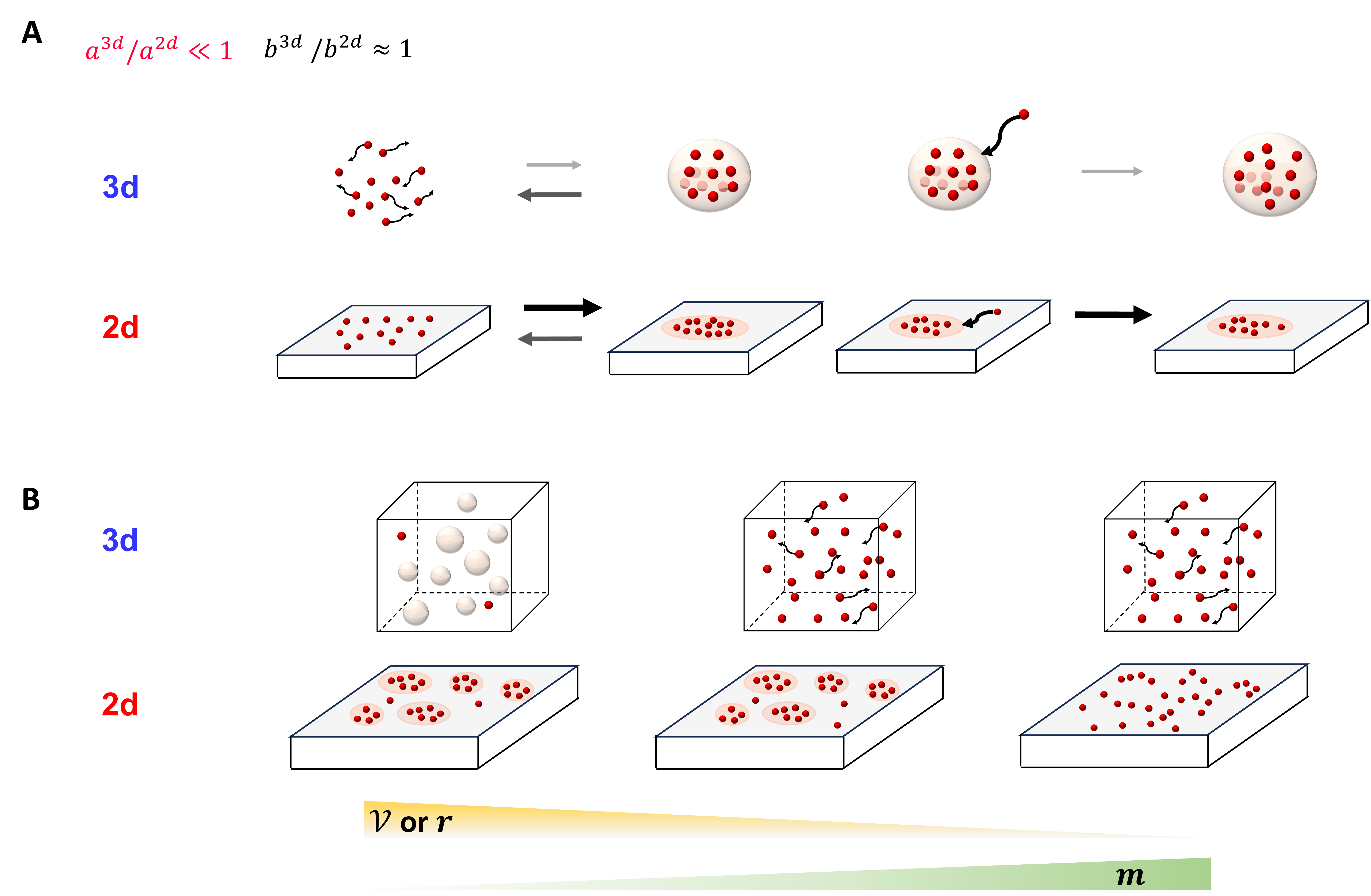}
    \caption{\textbf{Summary of the analysis of the rates}.
    \textbf{A}. Under identical settings, due to the difference between 2d and 3d in $a_1$ and $a_k$'s, droplet formation and coarsening more quickly happen in 2d than 3d. \textbf{B}. We schematically summarize the main results of this paper. The higher the viscosity (or the hydrodynamic radius) both 2d and 3d have more droplets. However 2d and 3d have different responses to changes of $\mathcal V$ or $r$. Similarly, while both have fewer droplets with a higher $m$, they have different responses to changes of $m$. }
\label{fig:scheme}
\end{figure}

\section{Results}\label{sec:main results}
In this section, we provide four main  results for qualitative comparison between 2d-LLPS and 3d-LLPS. We used plots of the stationary distribution $\pi$ and stochastic simulation algorithms to show our results and validate the analyses in Section \ref{sec:intuition}. In Section \ref{sec:analysis}, we provide proofs verifying our results. Without loss of generality, we use a unit radius of the system domain, $R=1$, for all the following simulations (except for Figure \ref{fig:free energy and time trajectory}C).

\subsection{Reproduction of LLPS}\label{sec:reproduction}
In this section, we show the consistency of our model with certain aspects of thermodynamics. Thermodynamic analysis validates that when the total concentration of the system is $\phi^*$, the system admits coexistence of dilute phases of concentration $\phi_1$ and dense phases of concentration $\phi_2$ rather than a single phase of concentration $\phi^*$ (Figure \ref{fig:free energy and time trajectory}A). This is due to the concave region of the free energy function that implies the convex combination of the free energies at $\phi_1$ and $\phi_2$ is less than the free energy at the concentration $\phi^*$. That is, $sF(\phi_1)+(1-s)F(\phi_2)<F(\phi^*)$ for $s$ such that $s\phi_1+(1-s)\phi_2=\phi^*)$ (for more details, see \cite{hyman2014liquid}). This induces phase separation, and the phase diagram is derived as Figure \ref{fig:free energy and time trajectory}B.
 This phase separation and the phase diagram can be reproduced with samples of our Markov chain $X(t)$  \eqref{eq:main crn} modeled with the reaction rates defined in Section \ref{sec:rates}.

\begin{figure}[ht]
    \centering
\includegraphics[width=0.8\textwidth]{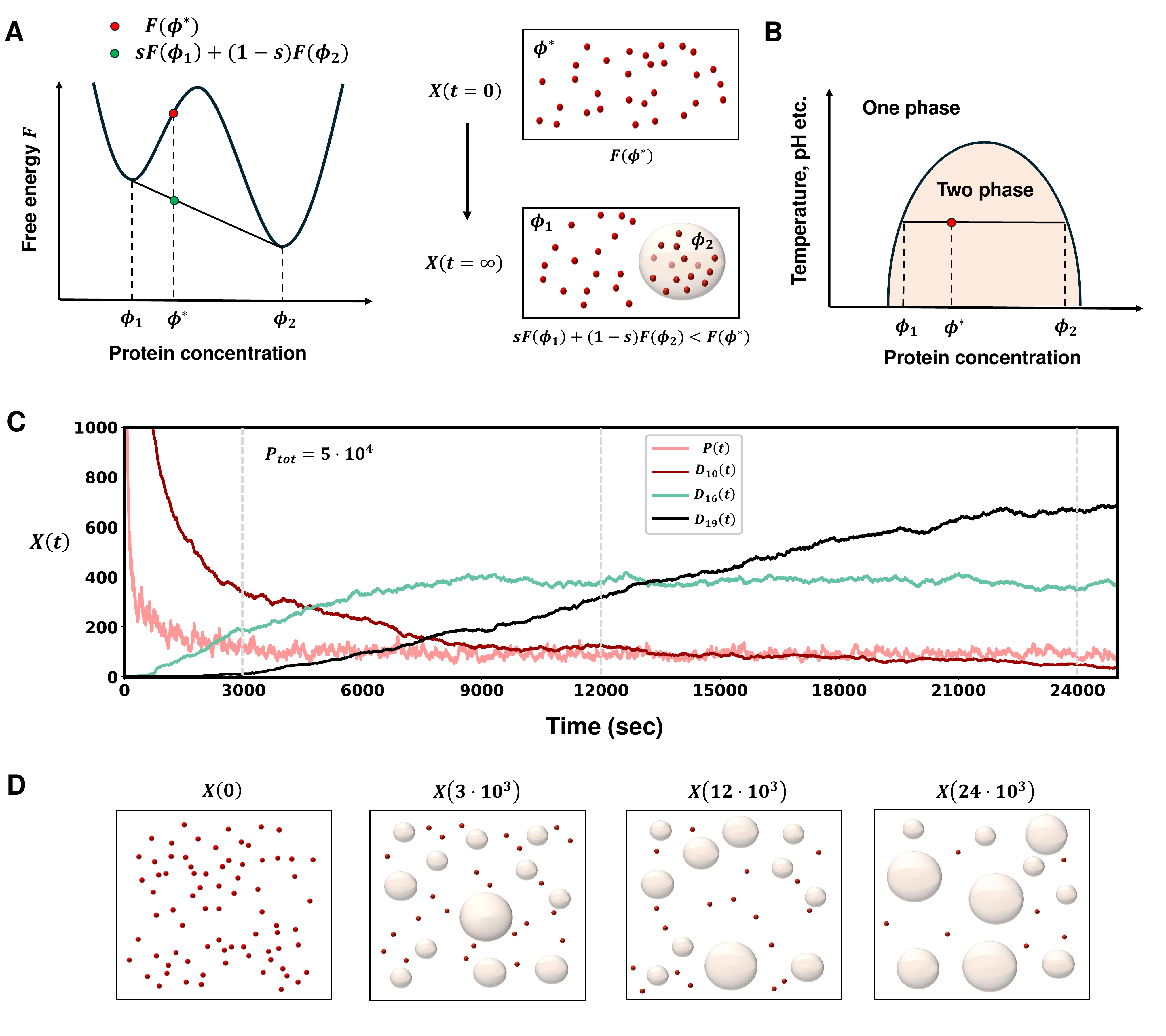}
    \caption{\textbf{Reproduce of phase separation.}
    \textbf{A} and \textbf{B}. Free energy explains phase separation and the phase diagram.  \textbf{C}. $x(t)=(p(t),d_m(t),\dots,d_L(t))$, a time trajectory of $X(t)$ in 2d, shows droplet coarsening and Ostwald ripening. \textbf{D}.  Cartoons of the snapshots of $X(t)$ at four time points.}
\label{fig:free energy and time trajectory}
\end{figure}

We first visualize samples of $X(t)$ in 2d to show how our model can describe phase separation. Once an initial condition $X(0)=(P_{tot},0,0,\dots,0)$ is fixed, where $P_{tot}$ is the initial protein count, we sample a single trajectory in time $x(t)=(p(t),d_m(t),\dots,d_L(t))$ using the statistically exact Gillespie algorithm \cite{Gill76}. The plot of $x(t)$ in Figure \ref{fig:free energy and time trajectory}C (with realistic parameters \cite{snead2022membrane}, see Appendix \ref{app:parameters}) shows the coarsening of the droplets as time passes. Notably, the number of smaller droplets decreases while the number of bigger droplets increases. Hence plots also visualize Ostwald ripening as individual proteins dissociate from a smaller droplet and join a bigger droplet (see Appendix \ref{app:ostwald} for more details). Then we display the sampled state at four time points ($t=0, 3\cdot 10^3, 12\cdot 10^3$, and $24\cdot 10^3$) by randomly distributing the proteins and the droplets over the space, where the counts of the proteins and droplets are given by $x(t)$ (Figure \ref{fig:free energy and time trajectory}D).

Now, by computing the volume fraction of the dense phase (droplets) and the dilute phase (outside droplets), we reproduce the phase diagram. Let $V_{k,d}$ be the volume of $D_k$, the droplet of size $k$. Then
\begin{align*}
    V_{k,d}=\begin{cases}
        \pi r_{k,d}^2 =\pi r^2 k\quad &\quad \text{if $d=2$},\\
        \dfrac{4}{3}\pi r_{k,d}^3 =\dfrac{4}{3}\pi r^3 k  &\quad \text{if $d=3$},
    \end{cases}
\end{align*}
    Both the number of proteins in $D_k$ and the volume of $D_k$ grow linearly in $k$. Hence, the concentration of the proteins inside the droplet $D_k$ is the same for each $k$, which is consistent with the previous analyses for LLPS \cite{alberti2018LLPS, hyman2014liquid}.

First of all, denoting by $\mathbb E_\pi(P)$ and $\mathbb E_\pi(D_k)$ the expected numbers of the proteins on the dilute phase and the droplet of size $k$ 
 with respect to $\pi$ \eqref{eq:pi}, we define the average volume ratio of the droplets to the system size $V$ as
 \begin{align}\label{eq:vol ratio} 
   \rho&:= \frac{\sum_{k=m}^L V_{k,d} \mathbb E_\pi (D_k)}{V} =\begin{cases}
       \dfrac{\pi  r^2(P_{tot}- \mathbb E_\pi (P))}{V} \quad &\quad \text{if $d=2$} ,\\
       { }\\
      \dfrac{\pi  \frac{4}{3}r^3(P_{tot}-\mathbb E_\pi (P))}{V} \quad &\quad \text{if $d=3$} ,
   \end{cases}
\end{align}
where we used the conservation of the total protein counts such that 
\begin{align*} P_{tot}:=P(t)+\sum_{k=m}^L k D_k(t) \quad \text{for each time $t$}.
\end{align*}
The $\rho$ in \eqref{eq:vol ratio} can determine whether the system has either a single phase or two phases.
 We assume that the system has two phases if $\rho \in (\rho_1,\rho_2]$ for some $\rho_1 < \rho_2$ and has a single phase otherwise. That is, if droplets and proteins coexist with the volume fraction falling in the range, the system has two phases. Furthermore, since the viscosity of liquid droplets and the hydrodynamic radius are specific functions of thermodynamic temperature $T+273$ which is measured in Kelvin, we define the constant $\mathcal V$ as a linear function of $e^{1/(T+273)}$
  and $r$ as a linear function of $T+273$, as referenced in the literature \cite{sherwood1961Exponential_visco, baer2024stokes_radius} (a full description of $\mathcal V$ and $r$ as functions of $T$ is given in Appendix \ref{app:parameters}). In this setting, we regard $\rho$ as a function of the total protein count $P_{tot}$ and the temperature $T$.

To compute the volume fraction with $\pi$ \eqref{eq:pi}, the state space has to be identified. However, when the initial number of the proteins is high, the state space is too large to search numerically.

Therefore, we sample $10^2$ time trajectories, using the Gillespie algorithm \cite{Gillespie77} for 2d with up to $3\cdot 10^4$ reactions, and the tau-leaping algorithm \cite{cao2005tau_leaping} for 3d with up to $3.5\cdot 10^4$ reactions. Then we empirically compute $\rho$ using the samples. Supplementary plots show that the samples with $3\cdot 10^4$ reactions in 2d and $3.5\cdot 10^4$ reactions in 3d closely approximate the stationary state of the system in 2d and 3d, respectively (Figure \ref{fig:gillespie stability}).

Figure \ref{fig:phase diagram}B is the graph of $\rho$ as a function of $P_{tot}$ with different values of temperature $T$ (Celsius).  
With this $\rho$, the phase diagram was obtained from our model with $\rho_1= 4\cdot 10^{-2}$ and $\rho_2=0.6$, and it turns out to display the well-known concave curve (Figure \ref{fig:phase diagram}C). 

\begin{remark}
    The plateau of $\rho$, which appears for small values of $P_{tot}$, describes that droplets are not formed when the total concentration 
$P_{tot}/V$ is small. Hence our model also reproduces the threshold protein concentration for phase separation.
\end{remark}

\begin{remark}
    The numbers $4\cdot10^{-2}$ and $0.6$ in the criteria of two phases $\rho\in (4\cdot10^{-2}, 0.6]$ were determined arbitrarily. Especially, to obtain $\rho>0.6$, a large total number of the proteins $P_{tot}$ is required leading to high computational costs. Nonetheless, the trend of the volume fraction curve in Figure \ref{fig:phase diagram}B guarantees reproduction of the phase diagram using the criteria of two phases $\rho \in (\rho_1, \rho_2]$ even if $\rho_2$ is selected sufficiently large. 
\end{remark}

\begin{remark}
Note that Case et al. \cite{case2022membranes} mentioned that two-dimensional spaces, such as cell membranes, shift the phase diagram to the left, promoting nucleation. Our model can also reproduce such a shift of the phase diagram in 2d-LLPS (Figure \ref{fig:phase diagram}C).
\end{remark}

\begin{figure}[t]
    \centering
\includegraphics[width=0.8\textwidth]{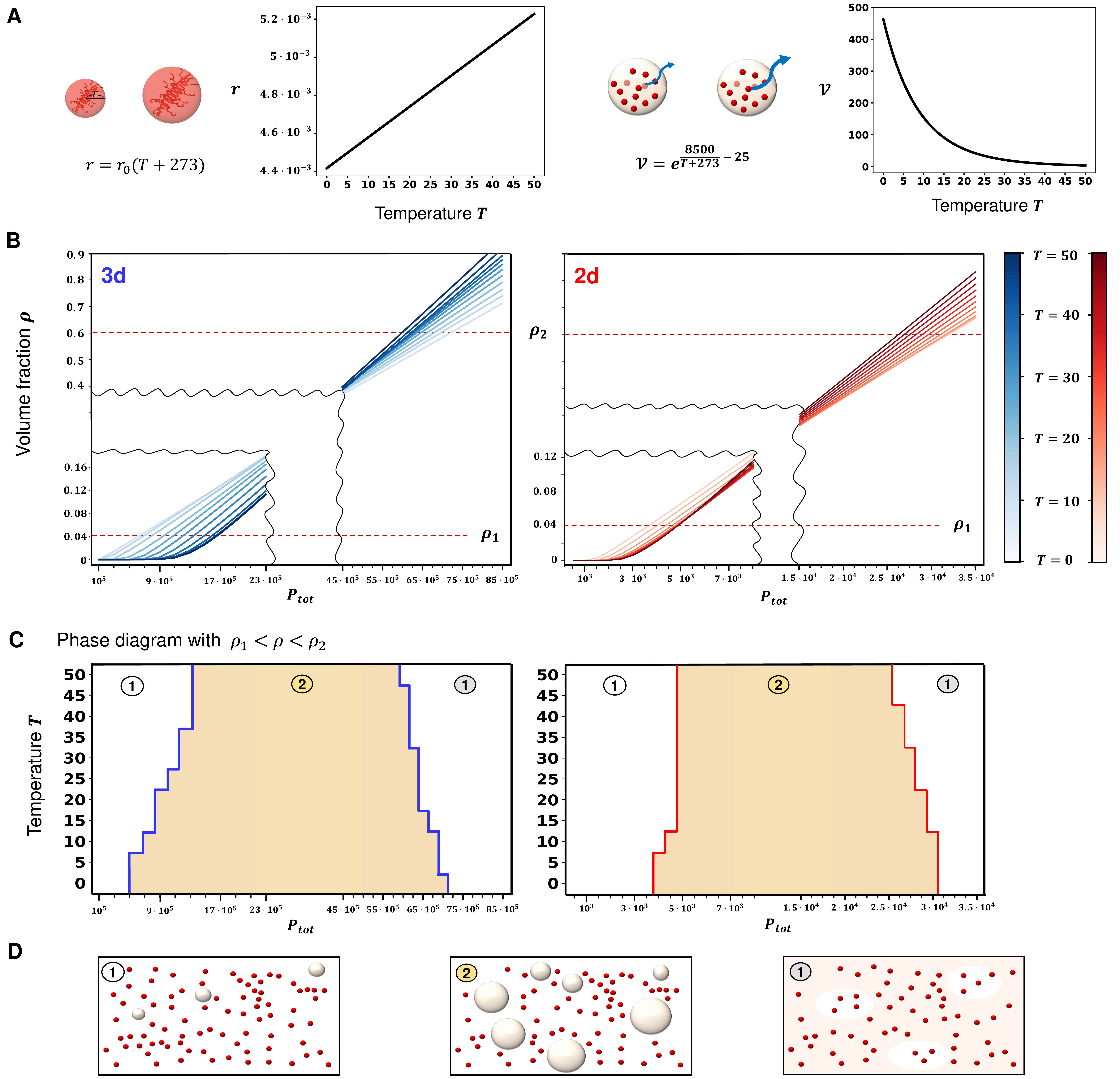}
    \caption{ \textbf{Reproduce Phase diagram which is from energy theory in physics.}\\
\textbf{A.} Temperature effect to the hydrodynamic radius $r$ and viscosity $\mathcal{V}$ with the graphs. \textbf{B.} The volume fraction \eqref{eq:vol ratio} as a function of $P_{tot}$ in 2d and 3d. $\rho_1$ and $\rho_2$ are the criteria for phase separation.  \textbf{C.} Phase diagrams in 2d and 3d. For example, a system at \textcircled{1} and \textcircled{2} will have a single phase and two phases, respectively, which are schematically illustrated in \textbf{D}.}
\label{fig:phase diagram}
\end{figure}

\subsection{Higher viscosity is required for LLPS in 3d.}\label{sec:result_valency }

\begin{figure}[htbp]
    \centering
\includegraphics[width=0.9\textwidth]{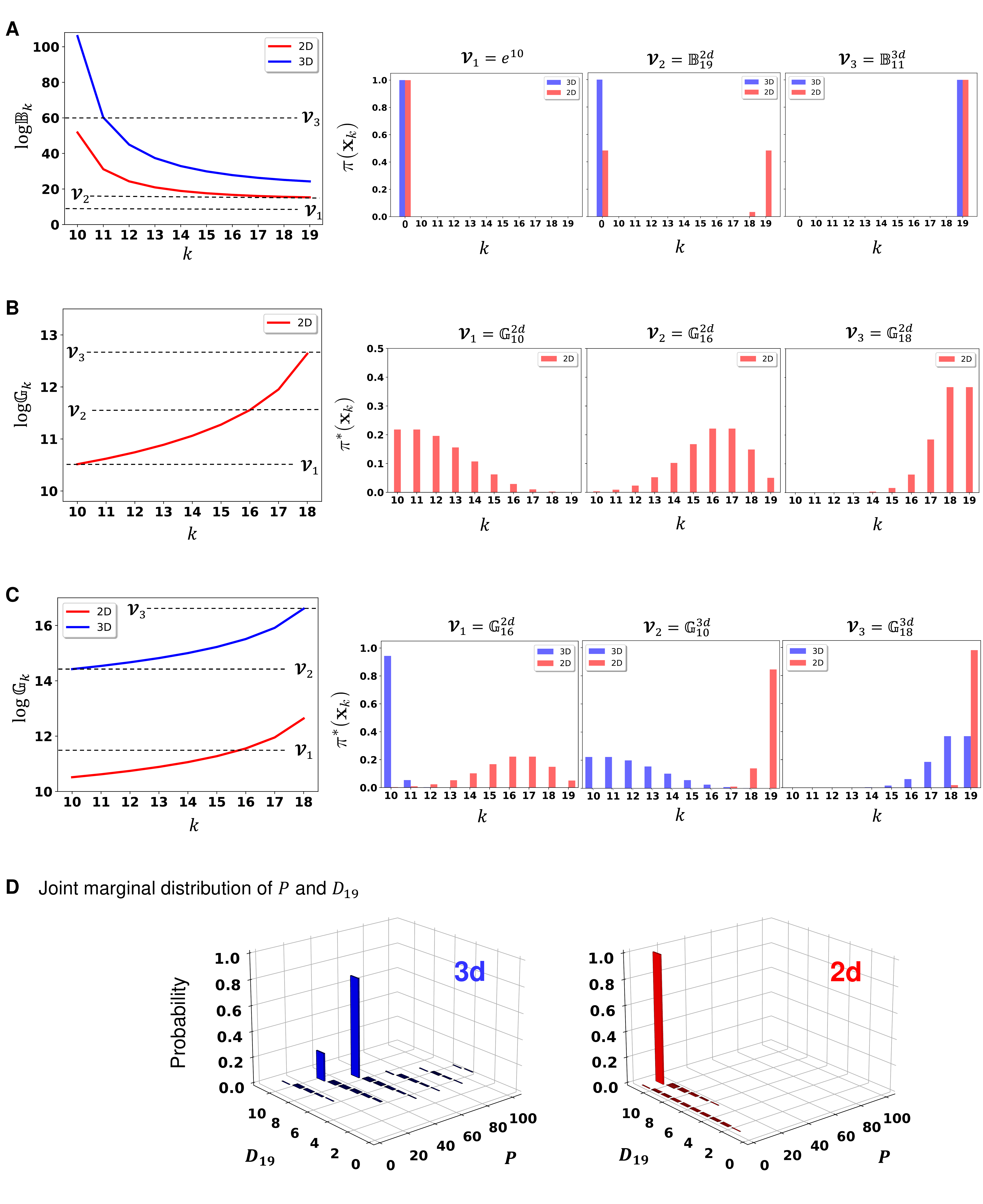}
    \caption{ \textbf{Effect of $\mathcal V$ in droplet formation.} 
    \textbf{A-C.} The log-scaled graph of $\mathbb B_k$ and $\mathbb G_k$ defined as \eqref{eq:B_k} and \eqref{eq:G_k}, respectively (left). The stationary distributions of $X$ with $(m,L) = (10,19)$ and $P_{tot} = L$ show the different probabilities of droplet formation and local maximums in 2d and 3d for certain values of $\mathcal V_i$'s (right). For \textbf{B} and \textbf{C}, for clearer visualization, we used $\pi^*(\x_k)=\pi(\x_k| P(0)\neq L)$, the stationary distribution conditioned on states consisting of at least one droplet. \textbf{D.} The joint marginal stationary distributions of $P$ and $D_{19}$ with $P_{tot}=200$.}
\label{fig:viscosity}
\end{figure}

In this section, we show that a higher viscosity is necessary for droplet formation in 3d than 2d by showing how the viscosity constant $\mathcal{V}$ alters the shape of the stationary distributions $\pi$ of 2d-LLPS and 3d-LLPS. To visualize the stationary distributions, we set the initial protein count $P_{tot}=L$ and $2m>L$ so that the state space $\mathbb S$  of the Markov chain $X(t)$ is 
\begin{align}\label{eq:state space}
    \mathbb S=&\{(L,0,\dots,0),(L-m,1,0,\dots,0), (L-m-1,0,1,0,\dots,0),\dots,(0,0,\dots,1)\},
\end{align}
and hence it can be linearly aligned.
We denote these states by 
\begin{align*}
     \mathbf x_0=(L,0,0,\dots,0), \mathbf x_m=(L-m,1,0,\dots, 0), \mathbf x_{m+1}=(L-m-1,0,1,0,\dots,0),\quad \cdots \\
     \mathbf x_L=(0,0,0,\dots,1).
\end{align*}
We are interested in finding a range of $\mathcal V$ for which the stationary distribution has a peak at a two phases state $\x_k$ for some $k\ge m$, where a droplet is formed (Figure \ref{fig:viscosity}A).  
To do that, we find $\mathbb B_k$'s such that 
\begin{equation}\label{eq:B_k}
	\mathbb{B}_{k} \leq \mathcal{V} < \mathbb{B}_{k-1} \quad  \text{if and only if} \quad \pi (\mathbf x_0 ) \leq \pi (\mathbf x_j ) 
\end{equation}
for any $k \leq j \leq L$. Using the closed forms of $\mathbb B_k$ for both 2d-LLPS and 3d-LLPS, it turns out that higher $\mathcal V$ is required for $\pi$ to have a peak at $\x_k$ for some $k$ in 3d than 2d (Figure \ref{fig:viscosity}A).  As mentioned in Section \ref{sec:intuition}, this result is not surprising because the reaction rates $a_1$ and $a_k$'s are bigger in a 2d space while $b_k$'s are comparable (Figure \ref{fig:overall}G).

The closed form of $\pi$ leads us to other interesting analyses about the relation between $\mathcal V$ and the droplet size distribution. For example, we investigated the range of $\mathcal V$ for which $\pi$ has a local maximum at the state $\mathbf x_k$ for $m\le  k\le L$ (Figure \ref{fig:viscosity}B,C). 
  Let $\mathbb{G}_k$ denote the viscosity constant such that  
  \begin{align}\label{eq:G_k}
      \mathbb G_{k-1} \le \mathcal V\le \mathbb G_k \quad  \text{if and only if} \quad 
 \pi(\mathbf x_j)  \le \pi(\mathbf x_k)
  \end{align} for all $m\le j \le L$.
  3d-LLPS also required a higher $\mathcal V$  to have a local maximum at $\x_k$ for some $k\ge m$ than 2d-LLPS (Figure \ref{fig:viscosity}C). 
$\mathbb G_k$ can also be related to an important experimentally observed phenomenon, droplet arrest (the growth of small or mid-size of droplets is paused) \cite{snead2022membrane}. 
 Precise calculations about $\mathbb B_k$ and $\mathbb G_k$ are given in Section \ref{sec:theoretical viscosity}.

 For general state spaces with larger $P_{tot}$, we can check the same trend. With $P_{tot}=200$, the selection of marginal stationary distributions for the counts of proteins $P(t)$ and the count of the largest droplet $D_L(t)$ show that a smaller number of the largest droplets are produced in 3d than in 2d (Figure \ref{fig:viscosity}D).

\subsection{For large thresholds, droplets can be formed in 2d but not in 3d}
\label{sec:result_m}
 It was claimed that membranes reduce the threshold concentration for phase separation  \cite{snead2022membrane, case2022membranes, snead2019control}.
Snead et al.\  \cite{snead2022membrane}
showed that anchoring proteins onto membranes may induce a shift of the threshold concentration for phase separation compared to the threshold of 3d LLPS. Motivated by these experimental findings, in this section we study the effect of $m$ (the threshold number of proteins for forming droplets) on LLPS. By varying $m$, we measure the probability of the state where droplets are formed. We used identical parameters for 2d-LLPS and 3d-LLPS. We identified a range of $m$ for which the droplet formation probability is nearly one in 2d. However, in contrast, it is nearly zero for 3d. We prove this mathematically in Section \ref{sec:analysis m}

Let $\x_0$ denote  $(P_{tot},0,\dots,0)$, the state without droplets.
 We used the probability $1-\pi(\mathbf x_0)$ as a function of $m$ to measure the probability that proteins form droplets.
 For either a large or a small $m$, both the probability $1-\pi(\xn)$ is either nearly $0$ or $1$ for both 2d-LLPS and 3d-LLPS. Interestingly, for an $m$ in the intermediate range,  $1-\pi(\xn)$ in 2d can be much greater than in 3d (Figure \ref{fig:threshold}). We prove the existence of such a range $m$ in Section \ref{sec:analysis m}.
This difference mainly arises from the rate constant $a_1$ as we highlight in Section \ref{sec:intuition}. In 2d, $m$ copies of proteins closely gather more frequently than 3d. This interpretation is consistent with the claim of the experimental previous study \cite{snead2022membrane}, the enhancement of the local concentration of the proteins on membranes. 
 Thus we can imagine that cells use membranes to efficiently facilitate the formation of biomolecular condensates at lower costs \cite{snead2019control}.

Furthermore, Figure \ref{fig:threshold} shows that droplet formation probabilities $1-\pi(\xn)$ dramatically decrease around certain $m$ (indicated by circles) in both 2d and 3d.  
This indicates the sensitivity of droplet formation to the minimum number of proteins or nucleation barriers \cite{martin2021multi}.

\begin{figure}[ht]
    \centering
\includegraphics[width=0.8\textwidth]{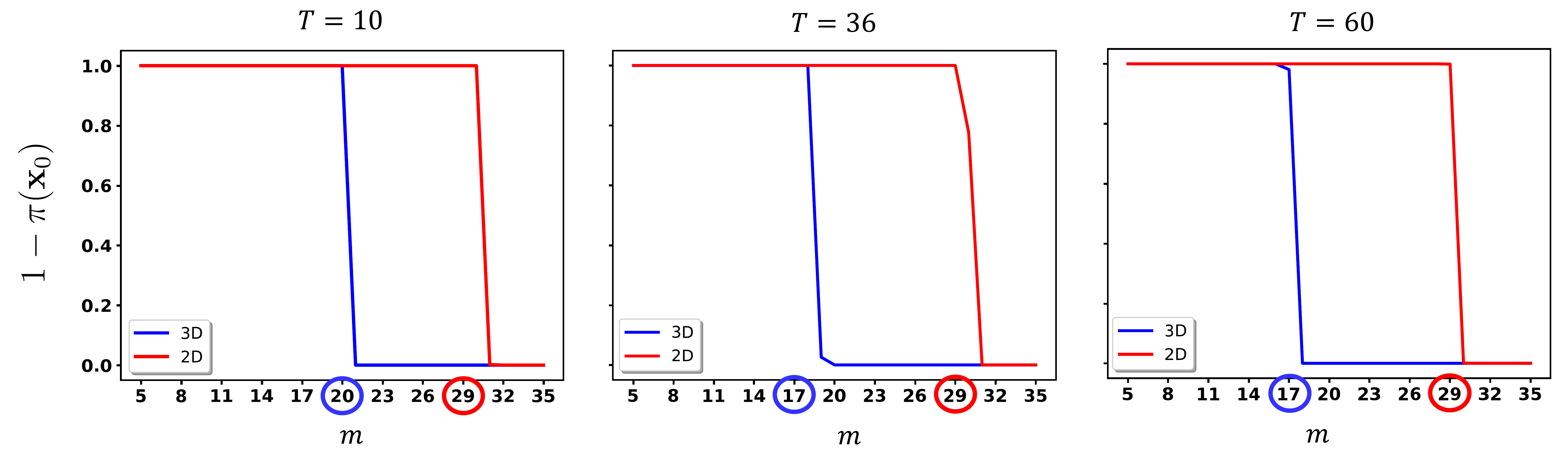}
    \caption{ \textbf{Effects of $m$.}
 The plots of the probabilities of droplet formation ($1-\pi(\xn)$) as functions of the threshold $m$ with three different choices of temperature $T$. For both 2d and 3d, there are critical $m$ (indicated with circles) where the probabilities dramatically drop.}
\label{fig:threshold}
\end{figure}

\subsection{Reduction of the hydrodynamic radius in 2d may not be significant}
\label{sec:result_r}
In Section \ref{sec:intuition}, we analyzed how much the coarsening rates $a_k$'s and the droplet formation rate $a_1$ are greater in 2d than in 3d. This leads to the main difference, and 2d-LLPS tends to have more droplets with higher chances in the long run. However, by examining the dependence of $a_1$ on $r$ \eqref{eq:a_1}, it can be predicted that anchoring a protein to a membrane surface can reduce the hydrodynamic radius (Figure \ref{fig:radius} A). This reduction, in turn, may inhibit droplet formation in 2d.
In this section, despite the reduction of $r$ in 2d, the probability of droplet formation is still higher in 2d than in 3d as long as the fold change in the hydrodynamic radius is not too large. We further analytically quantify the ratio between the hydrodynamic radii in 2d and 3d, at which the probability of droplet formation in 3d become larger than that in 2d.

Under the same setting of the state space $\mathbb S$ \eqref{eq:state space},
we first display the stationary distributions with different values of $r$. We denote by $r^{2d}$ and $r^{3d}$ the hydrodynamic radius of a protein in 2d and 3d, respectively. 
 We fix $r^{2d}=0.005$ for 2d-LLPS and set $r^{3d}_i$ for 3d-LLPS as $r^{3d}_1=2r^{2d}, r^{3d}_2=5r^{2d}$, and $r^{3d}_3=5.5 r^{2d}$.
 Interestingly, even though $r^{2d}<r^{3d}_1<r^{3d}_2$, the probability of droplet formation, $1-\pi(\xn)$, remains higher in 2d than in 3d for $r^{3d}_1$ and $r^{3d}_2$ (Figure \ref{fig:radius}B left and middle). 
 For $r^{3d}_3$, 2d and 3d have similar $1-\pi(\xn)$ (Figure \ref{fig:radius}B, right). Using the following relation, 
\begin{align}\label{eq:F(H)}
    F(\mathbb H):&= \log \left (\frac{\sum_{k=m}^L \pi^{3d}(\mathbf x_k)/ \pi^{3d}(\mathbf x_0)}{\sum_{k=m}^L\pi^{2d}(\mathbf x_k) / \pi^{2d}(\mathbf x_0)} \right )>0 
\end{align}
if and only if
\begin{align*}
    1-\pi^{3d}(\xn) > 1-\pi^{2d}(\xn),
\end{align*} we further see the fold changes of these probabilities by varing $\mathbb H$ such that $r^{3d}=\mathbb H r^{2d}$, where $\pi^{2d}$ and $\pi^{3d}$ denote the stationary distributions of $X(t)$ associated with 2d-LLPS and 3d-LLPS, respectively. The plot in Figure \ref{fig:radius}C, left shows that we can find the critical value for $\mathbb H=\mathbb H_0$ such that $F(\mathbb H_0)=0$ meaning that we have that $1-\pi^{3d}(\xn) > 1-\pi^{2d}(\xn)$ if and only if $\mathbb H>\mathbb H_0$.  
 
Using the closed form of $\pi$ \eqref{eq:pi}, we can also derive a closed form of $\mathbb H_k$ such that for each $k\ge m$, \begin{equation}\label{eq:H_k}
    \frac{\pi^{3d}(\mathbf x_k)/\pi^{3d}(\mathbf x_0)}{\pi^{2d}(\mathbf x_k)/\pi^{2d}(\mathbf x_0)}\ge 1 \quad \text{if and only if} \quad \mathbb{H} \ge \mathbb{H}_k.
\end{equation}
These $\mathbb H_k$'s guarantee a greater probability of droplet formation in 3d because $F(\mathbb H)>0$ if $\mathbb H\ge \mathbb H_k$ for all $k$ by \eqref{eq:F(H)}. On top of this, $\mathbb H_k$ turns out decreasing in $k$ (Figure \ref{fig:radius}C, right). Hence we have 
\begin{align*}
    1-\pi^{3d}(\xn) > 1-\pi^{2d}(\xn) \quad \text{if} \quad \mathbb H\ge \mathbb H_m.
\end{align*}
See Section \ref{sec:result_r} for the derivation of the closed form of $\mathbb H_k$.

\begin{figure}[ht]
    \centering
\includegraphics[width=0.8\textwidth]{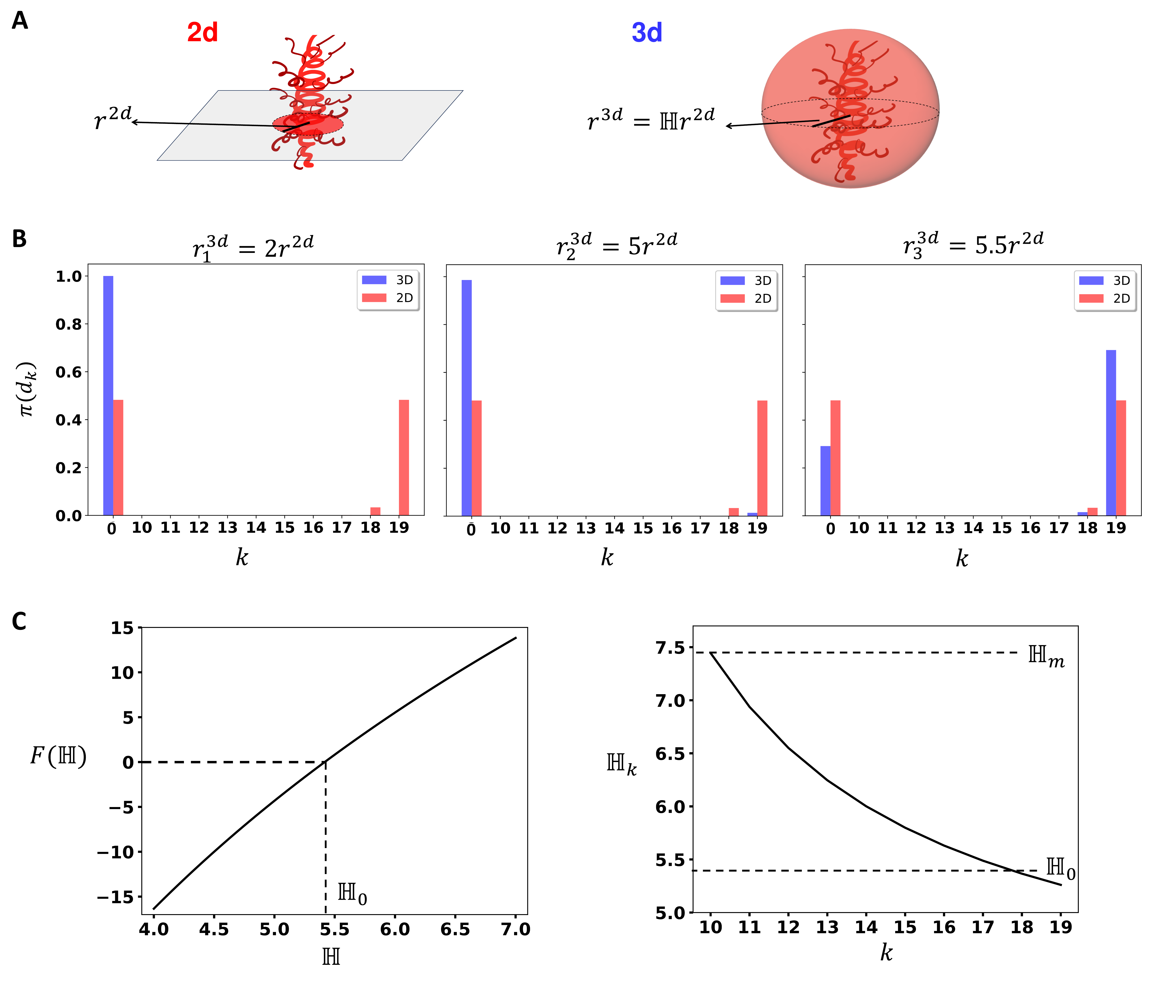}
    \caption{ 
\textbf{Effect of $r$.}
\textbf{A.} The hydrodynamic radii of a protein in 2d and 3d spaces. \textbf{B.} Stationary distributions with a fixed value of $r$ in 2d, and with different values of $r$ in 3d. \textbf{C.} (Left) The plot of $F(\mathbb H)$ in  \eqref{eq:F(H)} as a function of the ratio $\mathbb H=r^{3d}/r^{2d}$ between the hydrodynamic radii. $\mathbb H_0$ is the critical value such as $ 1-\pi^{3d}(\xn) > 1-\pi^{2d}(\xn)$ if and only if $\mathbb H>\mathbb H_0$.  (Right) $\mathbb H_k$ in \eqref{eq:H_k} is decreasing in $k$.}
\label{fig:radius}
\end{figure}

\section{Mathematical Analysis}\label{sec:analysis}
In this section, we validate all the results shown in Sections \ref{sec:result_valency }--\ref{sec:result_r} using the stationary distributions \eqref{eq:pi} of the stochastic model for LLPS. 
Let $\mathbb{S}_{\mathbf{x}}$ denote the closed communication class containing the initial state $\mathbf{x}$. We first provide the closed form of $Q_k$
 's shown in \eqref{eq:Q_k}, which will be used for analysis of the main results. 
 By using the closed form of $a_1$, $a_k$'s and $b_k$'s defined in Sections \ref{sec:a1}--\ref{sec:bi}, we have that for $m \leq k \leq L$,

\begin{align}\label{eq:Qk}
	Q_k &:=  \begin{cases} 
		\dfrac{\mathcal{T}_m}{\mathcal{D}} \bigg[\dfrac{m^{m-2} }{\Gamma(m-2)}  \bigg] \dfrac{\mathcal{V}^{k-m+1}}{\bigg[ y^{2k-2} \displaystyle \prod_{s=m}^{k-1} \ln (y^2/ s) \bigg]} & \quad \text{if $d=2$},\\
 \dfrac{\mathcal{T}_m}{\mathcal{D}}\bigg[ \bigg(\dfrac{3}{4}\bigg)^{m-1} \dfrac{  \sqrt{\pi}^{(m-1)} 2 m^{(3m-4)/3} }{ 3 \Gamma((3m-5)/2)} \bigg]  \dfrac{\mathcal{V}^{k-m+1} }{y^{3k-3}}
		\bigg( \dfrac{ m^{1/3} }{k^{1/3} 2^{k-m}} \bigg)  & \quad \text{if $d=3$},
	\end{cases}	
\end{align} where $y=1/r$ with the hydrodynamic radius of the protein is $r$ and $\mathcal{T}_m$ is defined in \ref{sec:a1}. $\mathcal{D}$ represents the diffusion coefficient of a single protein. We also adopt the convention that $\prod_{s=m}^{m-1} u_s = 1$ for any sequence $u_s$. We further assume that the radius of the domain $R=1$ with sufficiently small $r$  such that  $L^{1/2} r \ll 1$, and a fixed temperature $T$.

\subsection{Theoretical validation for Section \ref{sec:result_valency }}\label{sec:theoretical viscosity}
We demonstrate how the values of $\pi$ vary with the viscosity constant $\mathcal{V}$. For simplicity, we choose the number of initial proteins $P_{tot}=L$ and $L=2m-1$ for a fixed threshold number $m$ of proteins for forming droplets.  In this setting, the state space is as \eqref{eq:state space}, and we use the same notations $\mathbf x_k$'s for the states in this section. 

We first determine a range of viscosity with which the LLPS model has two modes in the stationary distribution. 

\begin{proposition}\label{prop:V for bimodal}
Suppose that $m\ge 5$, $P_{tot}=L$ and $L=2m-1$ in \eqref{eq:main crn}.  Then there exist $\{ \mathbb{B}_k\}_{m\leq k\leq L}$ such that 
\begin{equation*}
    \mathbb{B}_{k} \leq \mathcal{V} \quad  \text{ if and only if } \quad \pi (\mathbf x_0 ) \leq \pi (\mathbf x_k),
\end{equation*}
for both 2d and 3d. Furthermore, there exists $r_0$ such that $\{ \mathbb{B}_k\}_{m\leq k\leq L}$ is a decreasing sequence if $r < r_0$ in both 2d and 3d.
\end{proposition}

\begin{proof}
For any $m\ge 5$, by \eqref{eq:pi} and \eqref{eq:Qk}, we have that for any $m \leq  k \leq L$
\begin{equation*}
\dfrac{\pi(\mathbf x_{k})}{\pi( \mathbf x_{0})} = \dfrac{L!}{(L-k)!}Q_k \geq 1 
\end{equation*}
is equivalent to  $\mathbb{B}_k \le \mathcal V$, where 

\begin{align*}
	\mathbb{B}_k =  \begin{cases}
\bigg[\dfrac{\mathcal{D}}{\mathcal{T}_m} \dfrac{(L-k)!}{L!} \bigg\{\dfrac{ \Gamma(m-2)}{m^{m-2}}  \bigg\} y^{2k-2} 
		 \displaystyle \prod_{s=m}^{k-1} \ln (y^2/ s)  \bigg]^{1/(k-m+1)} &\quad \text{if $d=2$},\\
\bigg[ \dfrac{\mathcal{D}}{\mathcal{T}_m} \dfrac{(L-k)!}{L!} \bigg\{ \bigg(\dfrac{4}{3}\bigg)^{m-1} \bigg( \dfrac{ k^{1/3} 2^{k-m} }{m^{1/3}} \bigg)   \dfrac{ 3 \Gamma((3m-5)/2) }{ \sqrt{\pi}^{(m-1)} 2m^{(3m-4)/3}} \bigg\}   y^{3k-3}
		\bigg]^{1/(k-m+1)}  & \quad \text{if $d=3$}.
	\end{cases}	
\end{align*} 

Now we turn to show that $\mathbb B_k$ is a decreasing sequence in 2d. Note that for each $m \leq k < L$,
\begin{align*}
    \log \mathbb{B}_k - 	\log \mathbb{B}_{k+1} 
	&= C + \dfrac{2m - 4}{(k-m+1)(k-m+2)} \log y +\sum_{s=m}^{k-1}\dfrac{\log( \log(y^2/s))}{(k-m+1)(k-m+2)} + \dfrac{\log(y^2/k)}{k-m+2}, \notag
\end{align*}
where $C= \frac{\log \left ( \frac{\mathcal{D}}{\mathcal{T}_m L! } \frac{ \Gamma(m-2)}{m^{m-2}}\right ) }{ (k-m+1)(k-m+2) } + \frac{\log  (L-k)!  }{k-m+1} -  \frac{\log  (L-k-1)! }{k-m+2}$, which is independent of $y$.  Therefore, there exists $r_0$ such that if $y>\frac{1}{r_0}$, then $\mathbb{B}_k \geq \mathbb{B}_{k+1}$ for any $m \leq k < L$.  The proof for $\{ \mathbb{B}_k \}_{m \leq k \leq L}$ in 3d can be derived similarly such that
\begin{align}
    \log \mathbb{B}_k - \log \mathbb{B}_{k+1} 
	&= \bar{C} + \dfrac{3m - 6}{(k-m+1)(k-m+2)}\log y \notag 
\end{align} where $\bar{C} = \frac{\log \left( \frac{\mathcal{D}}{\mathcal{T}_m L! m^{1/3}}  \left ( ( \frac{4}{3})^{m-1}  \frac{ 3 \Gamma((3m-5)/2) }{ \sqrt{\pi}^{(m-1)} 2m^{(3m-4)/3}} \right )\right) }{ (k-m+1)(k-m+2) } + \frac{ \log \left( (L-k)! k^{1/3}2^{k-m} \right ) }{k-m+1}  - \frac{\log \left( (L-k-1)! (k+1)^{1/3}2^{k-m+1}  \right) }{k-m+2}$. Thus, there exists $r_0$ such that, if $y>\frac{1}{r_0}$, the inequality $\mathbb{B}_k \geq \mathbb{B}_{k+1}$ holds for all $m \leq k < L$.

\end{proof}

\begin{remark}
Proposition \ref{prop:V for bimodal} implies that 
\begin{equation*}
	\mathbb{B}_{k} \leq \mathcal{V} < \mathbb{B}_{k-1} \quad  \text{if and only if} \quad \pi (\mathbf x_0 ) \leq \pi (\mathbf x_j ) 
\end{equation*}
for any $k \leq j \leq L$. Thus, if $\mathcal V<\mathbb B_m$, then $\pi(\mathbf x_0) > \pi(\mathbf x_k)$ for any $m\le k \le L$, which means that there is only one mode at state $\mathbf x_0$.
\end{remark}

Now we turn to the ranges of viscosity with which a local maximum of $\pi$ is at $\mathbf x_k$. 
\begin{proposition}\label{prop:V for local max}
Suppose that $m \geq 5, P_{tot}=L$ and $L=2m-1$ in \eqref{eq:main crn}. Then there exist $\{ \mathbb{G}_k\}_{m\leq k<L}$ such that 
\begin{equation*}
    \mathbb{G}_k \leq \mathcal{V} \quad \text{ if and only if} \quad \pi (\mathbf x_k) \leq \pi(\mathbf x_{k+1}).
\end{equation*} Furthermore, there exists $r_0$ such that $\{ \mathbb{G}_k\}_{m\leq k< L}$ is an increasing sequence if $r < r_0$, for each case of $2d$ and $3d$.
\end{proposition}

\begin{proof} For any $m\geq 5$, by \eqref{eq:pi} and \eqref{eq:Qk}, we have that for any $m \leq  k < L$
\begin{equation*}
    \dfrac{\pi(\mathbf x_{k+1})}{\pi(\mathbf x_k)} = \dfrac{(L-k)!}{(L-k-1)!} \dfrac{Q_{k+1}}{Q_k} \geq 1
\end{equation*}
is equivalent to $\mathbb G_k \le \mathcal V$, where
\begin{equation} 
    \mathbb{G}_k = \begin{cases}
        \dfrac{y^2 \ln(y^2 / k )}{(L-k)} & \quad \text{if $d=2$},\\
        \dfrac{2y^3}{(L-k)} \left(\dfrac{k+1}{k}\right)^{1/3} &\quad \text{if $d=3$},
    \end{cases}
\end{equation} By simply calculation, in 2d, we find that if $y \geq  \sqrt{2(m-1) ( 1 + m^{-1} )^{m-1}}$  then for any $m \leq k < L$,

\begin{equation*}
    \dfrac{\mathbb{G}_{k+1}}{\mathbb{G}_{k}} = \dfrac{(L-k)}{(L-k-1)} \dfrac{\ln(y^2 / (k+1))}{\ln(y^2 / k)} \geq 1. \notag
\end{equation*} 
Hence for $r_0=1/\sqrt{(L-1)(1 + m^{-1})^{L-m}}$, the results holds in 2d.
Similarly, in 3d, we find that for any $5 \le m\le k<L$,

\begin{align}
\dfrac{\mathbb{G}_{k+1}}{\mathbb{G}_{k}} &= \dfrac{(L-k)}{(L-k-1)} \bigg[ \dfrac{k(k+2)}{(k+1)^2} \bigg]^{\frac{1}{3}} \geq \dfrac{m-1}{m-2} \bigg[ \dfrac{m(m+2)}{(m+1)^2} \bigg]^{\frac{1}{3}} \notag = \bigg( 1+ \dfrac{3m^4 - 4m^3 -5m^2 + 2m +8}{(m-2)^3(m+1)^2} \bigg)^{\frac{1}{3}}>1 \notag
\end{align} This implies that $\{\mathbb{G}_k\}_{m \leq k < L}$ is also an increasing sequence in 3d for any choice of sufficiently small $r$.
\end{proof}

\begin{remark}
Proposition \ref{prop:V for local max} implies that 
$\mathbb{G}_{k} \leq \mathcal{V} < \mathbb{G}_{k-1}$ if and only if
\begin{align}
      \pi(\mathbf x_j) \leq \pi(\mathbf x_k) \notag
\end{align}
for any $m \leq j \leq L$. Under this range of the viscosity, the stationary distribution has a local maximum at state $\mathbf x_k$.
Moreover, by \eqref{eq:G_k}, we can see that for sufficiently small $r$, $\mathbb G_k$ in 2d is smaller than $\mathbb G_k$ in 3d. 
\end{remark}

\subsection{Theoretical validation for Section \ref{sec:result_m}}\label{sec:analysis m}

We demonstrate the effect of the threshold  $m$ for droplet formation in 2d and 3d by using the value of the stationary distributions at state $\mathbf x_0=(P_{tot},0,0,\dots,0)$ that do not include droplets. In the context of comparison, we will also use $\pi^{2d}$ and $\pi^{3d}$ to denote the stationary distribution in 2d and 3d, respectively. Through the probabilities $\pi^{2d}(
\xn)$ and $\pi^{3d}(\xn)$, we prove that there exists a range of $m$ for which $\pi^{3d}(\mathbf x_0)$ is much greater than $\pi^{2d}(\mathbf x_0)$. We first derive an inequality for the ratio of the probability of forming no droplets between 2d and 3d.

\begin{proposition}\label{prop:m}
Suppose that $m\ge 5$. For fixed $L\ge m$ and $P_{tot} \ge m$, we have that for any $r<1$, 
\begin{equation*}
    \frac{\pi^{3d}(\xn)}{\pi^{2d}(\xn)} \ge \frac{1+y^{m-1}\sum_{\x \in \mathbb S_{\xn}/\{\xn\}} \frac{\pi^{3d}(\x)}{\pi^{3d}(\xn)}}{1+\sum_{\x \in \mathbb S_{\xn}/\{\xn\}} \frac{\pi^{3d}(\x)}{\pi^{3d}(\xn)}}.
\end{equation*}  where $y=1/r$. 
\end{proposition}

\begin{proof}
Let $m\ge 5$ be fixed. By \eqref{eq:pi} and \eqref{eq:Qk}, we can derive that for any state $\x=(\x_1,\x_m,\dots,\x_L)\in \mathbb{S}_{\xn}/ \{\xn\}$,  \begin{equation}\label{eq:ratio of pis}
\bigg(\dfrac{\pi^{2d}(\textbf{\textbf{x}})}{\pi^{2d}(\textbf{x}_0)} \bigg) \bigg \slash \bigg(\dfrac{\pi^{3d}(\textbf{x})}{\pi^{3d}(\textbf{x}_0)}\bigg) = \prod_{k=m}^{L} \bigg( \dfrac{Q_k^{2d}}{Q_k^{3d}}\bigg)^{\x_k} 
 \end{equation} 
 For any $m \leq k \leq L$, the ratio $Q_k$ of 2d and 3d are 

\begin{align}
Q_k^{2d} \slash Q_k^{3d} &= \Theta_m k^{1/3}\dfrac{y^{k-1}}{ \bigg[\displaystyle \prod_{s=m}^{k-1}\ln (y/ \sqrt{s}) \bigg] } = \Theta_m k^{1/3} \dfrac{y^{k-m}}{\bigg[\displaystyle \prod_{s=m}^{k-1}\ln (y/ \sqrt{s}) \bigg]} y^{m-1} \notag \\
&\geq \Theta_m  k^{1/3} \sqrt{\dfrac{(k-1)!}{(m-1)!}} y^{m-1}  \geq \Theta_m m^{1/3} y^{m-1}   \label{eq:ratio of Qs}
 \end{align} 
where $y=1/r$ and \begin{equation}\label{Theta_m}
    \Theta_m = \dfrac{\Gamma((3m-5)/2)}{\Gamma(m-2)} \bigg(\dfrac{4}{3} \bigg)^{m-2} \dfrac{2}{m \sqrt{\pi}^{(m-1)}}.
\end{equation}
For simplicity, we let $\ell = m-2$ and define $h(\ell)=\Theta_m m^{1/3}$. By the Bohr-Mollerup Theorem \cite{gubner2021gamma}, which gives asymptotics of the gamma function, we have

\begin{align}
    h(\ell) &\geq \frac{\bigg(\frac{3\ell+1}{2}\bigg)^{\frac{3\ell+1}{2}-\frac{1}{2}} e^{-\frac{3\ell+1}{2}}}{\ell^{\ell -\frac{1}{2}} e^{-\ell +\frac{1}{12 \ell}}} \bigg(\dfrac{4}{3}\bigg)^{\ell} \frac{2}{(\ell+2)^{2/3}} \dfrac{1}{\sqrt{\pi}^{^{\ell+1}}} \notag \\
    &> 
\dfrac{\left(\frac{3\ell +1}{2}\right)^{\frac{3\ell}{2}}}{\ell^{\ell - 1/2}} 
    \frac{1}{e^{\frac{\ell}{2}+ \frac{1}{12\ell}+\frac{1}{2}}} \bigg(\dfrac{2}{3}\bigg)^{\ell} \frac{1}{(\ell+2)^{2/3}}:=\bar h(\ell) \label{eq:h and h bar}
\end{align}
Note that $\bar h(\ell)$ is the non-negative function for all $\ell \geq 1$ with $\bar h(3) \geq 1$. By analyzing the derivative of $\log \bar h(\ell)$, we will show that $\bar h(\ell)$ is an increasing function as $\ell \geq 3$. Specifically, we have 

\begin{align}
\dfrac{d}{d\ell}\log \bar  h(\ell)  &= \bigg[ \log \bigg(1 + \dfrac{1}{3\ell} \bigg) + \dfrac{1}{12\ell^2} + \dfrac{1}{2\ell} \bigg]  + \dfrac{1}{2} \log\bigg(\dfrac{3\ell +1}{2} \bigg) - \dfrac{21\ell+22}{6(\ell+2)(3\ell+1)} \notag
\end{align} 
This given expression can be bounded below by sum of two increasing functions for $\ell \geq 3$ as follows \begin{align}
\dfrac{d}{d\ell}\log \bar h(\ell)  \geq 	 \dfrac{1}{2} \log\bigg(\dfrac{3\ell +1}{2} \bigg) + \dfrac{(-21\ell-22)}{6(\ell+2)(3\ell+1)} \notag
\end{align} 
Through some elementary calculations, we can show that $\dfrac{d}{d\ell}\log \bar h(\ell) >0$. This implies that $\log \bar h(\ell)$ is an increasing function of $\ell$ for $\ell \geq 3$. Consequently, we have that $\bar h(\ell)$ is an increasing function, which finally implies $h(\ell)\geq  1$ for any $\ell \ge 3$ by \eqref{eq:h and h bar}. 

Since for each $\x=(\x_1,\x_m,\dots,\x_L)\in \mathbb S_{\xn}/\{\xn\}$, it must hold that $\x_k \ge 1$ for at least one $k\geq m$. Therefore by \eqref{eq:ratio of pis} and \eqref{eq:ratio of Qs}
\begin{equation}\label{eq:ratio ineq}
\dfrac{\pi^{2d}(\textbf{\textbf{x}})}{\pi^{2d}(\textbf{x}_0)}  
 \ge y^{m-1}\dfrac{\pi^{3d}(\textbf{x})}{\pi^{3d}(\textbf{x}_0)}  
\end{equation}
for any $\textbf{x} \in \mathbb{S}_{\textbf{x}_0}/\{\xn\}$. We now establish the following equality such that 
\begin{align*}
\sum_{\textbf{x} \in \mathbb{S}_{\textbf{x}_0} } \pi( \textbf{x}) &= \pi( \textbf{x}_0) \sum_{\textbf{x} \in \mathbb{S}_{\textbf{x}_0} } \dfrac{\pi(\textbf{x})}{\pi(\textbf{x}_0)} = \pi( \textbf{x}_0) \left(1+\sum_{\textbf{x} \in \mathbb{S}_{\textbf{x}_0}/\{\xn\} } \dfrac{\pi(\textbf{x})}{\pi(\textbf{x}_0)} \right) = 1 \notag
\end{align*} So, we conclude that by \eqref{eq:ratio ineq} 

\begin{align*}
    \frac{\pi^{3d}(\xn)}{\pi^{2d}(\xn)} &= \frac{1+\sum_{\textbf{x} \in \mathbb{S}_{\xn}/\{\xn\} } \dfrac{\pi^{2d}(\x)}{\pi^{2d}(\xn)}}{1+\sum_{\x \in \mathbb{S}_{\xn}/\{\xn\} } \dfrac{\pi^{3d}(\x)}{\pi^{3d}(\xn)}} \ge 
   \frac{1+y^{m-1}\sum_{\x \in \mathbb S_{\xn}/\{\xn\}} \dfrac{\pi^{3d}(\x)}{\pi^{3d}(\xn)}}{1+\sum_{\x \in \mathbb S_{\xn}/\{\xn\}} \dfrac{\pi^{3d}(\x)}{\pi^{3d}(\xn)}}.
\end{align*}
\end{proof}

\begin{remark}
Now we show the existence of a range of $m$ where the probability of forming no droplets is significantly different between 2d and 3d in the following remark. To highlight the dependent on $m$, in this section we denote by $\pi^{2d}_m$ and $\pi^{3d}_m$ the stationary distribution of 2d-LLPS and 3d-LLPS, respectively.

Let $5 \le m_0\le L$ such that 
\begin{align}
    \sum_{\x \in \mathbb {S}^{m_0}_{\xn}/\{\xn\}} \frac{\pi_{m_0}^{3d}(\x)}{\pi_{m_0}^{3d}(\xn)} \ge \alpha >0 \label{eq:condition on pi3}
\end{align} for some $\alpha>0$ where $\mathbb{S}^{m_0}_{\xn}$ is the state space for given $m_0$.
Since $\frac{\pi_m(\x)}{\pi_m(\xn)}$ is obviously decreasing with respect to $m$ for any $\x \in \mathbb S^{m}_{\xn}/\{\xn\}$ and for both 2d and 3d, by Proposition \eqref{prop:m} we have that for any $m\in[5,m_0]$
\begin{equation*}
   \frac{\pi_m^{3d}(\xn)}{\pi_m^{2d}(\xn)} 
   \ge 
   \frac{1+y^{m-1}\sum_{\x \in \mathbb S^m_{\xn}/\{\xn\}} \frac{\pi_m^{3d}(\x)}{\pi_m^{3d}(\xn)}}{1+\sum_{\x \in \mathbb S^m_{\xn}/\{\xn\}} \frac{\pi_m^{3d}(\x)}{\pi_m^{3d}(\xn)}} \ge \frac{1+y^{m-1}\alpha}{1+\alpha},
\end{equation*}
here for the second inequality we used that the function $f(z)=\frac{1+\beta z}{1+z}$ is increasing for $z>0$ when $\beta>1$.
For instance, suppose that there exists $m_0$ such that \eqref{eq:condition on pi3} holds with $\alpha=1/2$. This roughly means that $m_0$ is not too big so that $\pi_{m_0}(\xn)$ is relative higher than $\pi_{m_0}(\x)$ for $\x \in \mathbb S^{m_0}_{\xn}/\{\xn\}$. In this case, Proposition \ref{prop:m} implies that for each $m\in [5,m_0]$, we have
\begin{equation*}
    \frac{\pi_m^{3d}(\xn)}{\pi_m^{2d}(\xn)} \ge \frac{1+0.5 y^{m-1}}{1.5},
\end{equation*} 
where $\frac{1+0.5 y^{m-1}}{1.5}$ is a large number if $y=1/r$ is sufficiently large.
\end{remark}

\subsection{Theoretical validation for Section \ref{sec:result_r}}

We explore a sufficient condition for the fold-change constant $\mathbb{H} \geq 1$, which enhances the probability of forming droplets in 3d, when the radius is reduced by anchoring a protein in 2d. This condition is defined by the ratio $r^{3d}/r^{2d} = \mathbb{H}$, where $r^{2d}$ and $r^{3d}$ represent the hydrodynamic radii of individual proteins in 2d and 3d, respectively.

\begin{proposition}\label{prop:r}
Suppose that $m\ge 5$ and $m \leq P_{tot}$.  Then for each $r^{2d}$, there exists  $\mathbb H^*$ such that 
\begin{align}
    1- \pi^{3d}( \mathbf{x}_0 ) \ge 1-\pi^{2d}( \mathbf{x}_0) \notag
\end{align}
if $\mathbb {H} \geq \mathbb {H}^* $, where $\mathbf x_0=(P_{tot},0,\dots,0)$.
\end{proposition}

\begin{proof}
By \eqref{eq:Qk}, for any $m \leq k \leq L$, we have
\begin{equation}
        Q_k^{3d}/ Q_k^{2d} =\dfrac{\mathbb{H}^{3k-3}}{\Theta_m \alpha_k} \geq 1
\end{equation} where $\alpha_k = k^{1/3} \left[ \frac{y^{k-1}}{\prod_{s=m}^{k-1} \ln(y/\sqrt{s})} \right]$ with $y=1/r^{2d}$ and $\Theta_m$ is defined as \eqref{Theta_m}. For each $m \leq k \leq L$, we define the sequence $\mathbb{H}_k = (\Theta_m \alpha_k)^{1/(3k-3)}$, which satisfies the following equivalence condition \begin{equation}
        \mathbb{H} \ge \mathbb{H}_k \quad \text{if and only if} \quad Q_k^{3d} /  Q_k^{2d}  \geq 1.
\end{equation} Let $\mathbb{H}^* := \max_{m\leq k\leq L} {\mathbb{H}_k}$. By the definition of  $\pi$ \eqref{eq:pi}, \begin{equation*}
    1- \pi^{3d}( \mathbf{x}_0 ) \ge 1-\pi^{2d}( \mathbf{x}_0) \notag
\end{equation*} if $\mathbb{H} \geq \mathbb{H}^*$ where $\mathbf x_0=(P_{tot},0,\dots,0)$.
\end{proof}

\begin{remark}
We found that there exists $r_0$ such that $\{\mathbb{H}_k\}_{m \leq k \leq L}$ is a decreasing sequence if $r^{2d} < r_0$ (Figure \ref{fig:radius}C). For any $m \leq k < L$, we show that $\log(\mathbb{H}_k/ \mathbb{H}_{k+1})$ is non-negative as follows

\begin{align*}
		\log \mathbb{H}_k - \log  \mathbb{H}_{k+1} &= \Theta_1 + \frac{ (k-1) \log ( \log (y/\sqrt{k})) - \sum_{s=m}^{k-1} \log ( \log (y/\sqrt{s})) }{3k(k-1)} \\
  &\geq \Theta_1 + \frac{ \bigg[ \log \bigg( \dfrac{\log (y/\sqrt{k})}{\log (y/\sqrt{m})} \bigg) + (m-1)\log \log (y/\sqrt{k})\bigg]}{3k(k-1)} 
\end{align*} where $\Theta_1 = \frac{\log \Theta_m^3 + \log [ k(k/k+1)^{k-1}]}{9k(k-1)}$ and $y=1/r^{2d}$.
Since, for any $\alpha, \beta > 0$, $\displaystyle \lim_{y \to \infty}  \dfrac{\log (y/\alpha)}{\log (y/\beta)} = 1$, there exist $r_0$ such that $\mathbb{H}_k \geq \mathbb{H}_{k+1}$ for any $m \leq k \leq L$ if $y > \frac{1}{r_0}$, i.e., $\mathbb{H}^* = \max_{m \leq k \leq L} \mathbb{H}_k = \mathbb{H}_m$.  

Consequently, for each $\ell \ge k$
\begin{align*}
  \dfrac{Q_\ell^{3d}}{Q_\ell^{2d}} = \frac{\pi^{3d}(\mathbf{x}_\ell)/\pi^{3d}(\mathbf{x}_0)}{\pi^{2d}(\mathbf{x}_\ell)/\pi^{2d}(\mathbf{x}_0)} \geq 1 \text{ if and only if } \mathbb H \ge \mathbb H_{k}.
\end{align*}
In the simple state space \eqref{eq:state space}, if $\mathbb H=\mathbb H_L$, then $\pi^{3d}(\x_L)/\pi^{3d}(\xn)$, the probability of forming the largest droplet in 3d relatively to the probability of no droplets, is bigger than or equal to  $\pi^{2d}(\x_L)/\pi^{2d}(\xn)$. However for $\pi^{3d}(\x_L)/\pi^{3d}(\xn)\ge \pi^{2d}(\x_L)/\pi^{2d}(\xn)$, a bigger reduction of the hydrodynamic radius in 2d is needed as it holds only if $\mathbb H\ge \mathbb H_m> \mathbb H_L$.
\end{remark}

\section{Conclusion}

We used a reaction network and the associated Markov chains to study how spatial dimension affects LLPS. We set the rate constants using the concepts of mean first passage times and the generalized Smoluchowski reaction kinetics. These rate constants capture spatial dimensional effects, and they further reflect the physical influence of temperature on protein interaction range and viscosity in hydrodynamics. Using chemical reaction network theory, we obtained a closed form of the stationary distribution and revealed qualitative differences between 2d-LLPS and 3d-LLPS using this closed form.

Our model successfully reproduces the phase diagram of LLPS as predicted by free energy. Building on this validation, we performed an analytical and numerical investigation into viscosity in both 2d and 3d. This investigation shows that 2d-LLPS can form droplets even at lower viscosities compared to 3d-LLPS. Further, there exists a range of the threshold number of proteins required for droplets formation in which 2d-LLPS has a much higher probability of forming droplets than 3d-LLPS. This may provide a reason why cells utilize 2d spaces such as ER membranes for LLPS. Finally, considering the effect of the hydrodynamic radius of proteins, our paper identifies the ratio of the radii between 2d and 3d for which 3d-LLPS can have a similar number of droplets compared to 2d systems, and this result is supported by an analytical proof.

The Markov model we proposed is based on the first passage times of diffusing particles. While we primarily analyzed the stationary distribution of the model, there are many avenues for future work analyzing other aspects of the model. For instance, one can use chemical reaction network theory and present the Markov process using the random time representation \cite{Kurtz72} to study the diffusion limit and the fluid limit of the model under the volume scaling and time scaling in future studies.
 Furthermore, the random-time representation and the Gillespie algorithm can be also employed to explore the transient dynamic of LLPS such as quasi-stationary behaviors and the pre-equilibrium behaviors. As such, our Markovian chemical reaction network theory of LLPS offers a new framework 
 for studying a variety of microscopic (or mesoscopic) perspectives on LLPS.

\vspace{1cm}

\appendix

\textbf{\Large{Appendix}}

\section{Modeling details}\label{app:model details}
In this section, we give more details pertaining to modeling LLPS with the stochastically modeled reaction network \eqref{eq:main crn}. 

\subsection{The minimum number of proteins for droplet formation}\label{sub_app:threshold}
The minimal number of proteins, $m$, for droplet formation is experimentally and theoretically validated in \cite{martin2021multi, qin2022nucleation_barrier,iida2024nucleation_barrier}. In \cite{martin2021multi} the authors used the condition of zero flux to derive the critical number of proteins to form the nucleation barrier.
We can also validate the existence of the minimal number with our Markov model and the volume fraction \eqref{eq:vol ratio}. As shown in Figure \ref{fig:m=3} (right), the volume fraction with $m=3$ (that is, a droplet can be formed with 3 proteins) immediately increases when $P_{tot}$ increases, as opposed to the case of $m=10$ displayed in Figure \ref{fig:m=3} (left), where the plateau of the volume fraction characterizes the existence of the threshold protein concentration for phase separation. This indicates that if an arbitrarily small number (such as $m\le 3$) of proteins can form droplets, then there is no threshold concentration of the protein to form droplets. Hence the condition of $m\le 3$ fails to capture the key feature of LLPS.

\begin{figure}[ht]
    \centering
\includegraphics[width=0.9\textwidth]{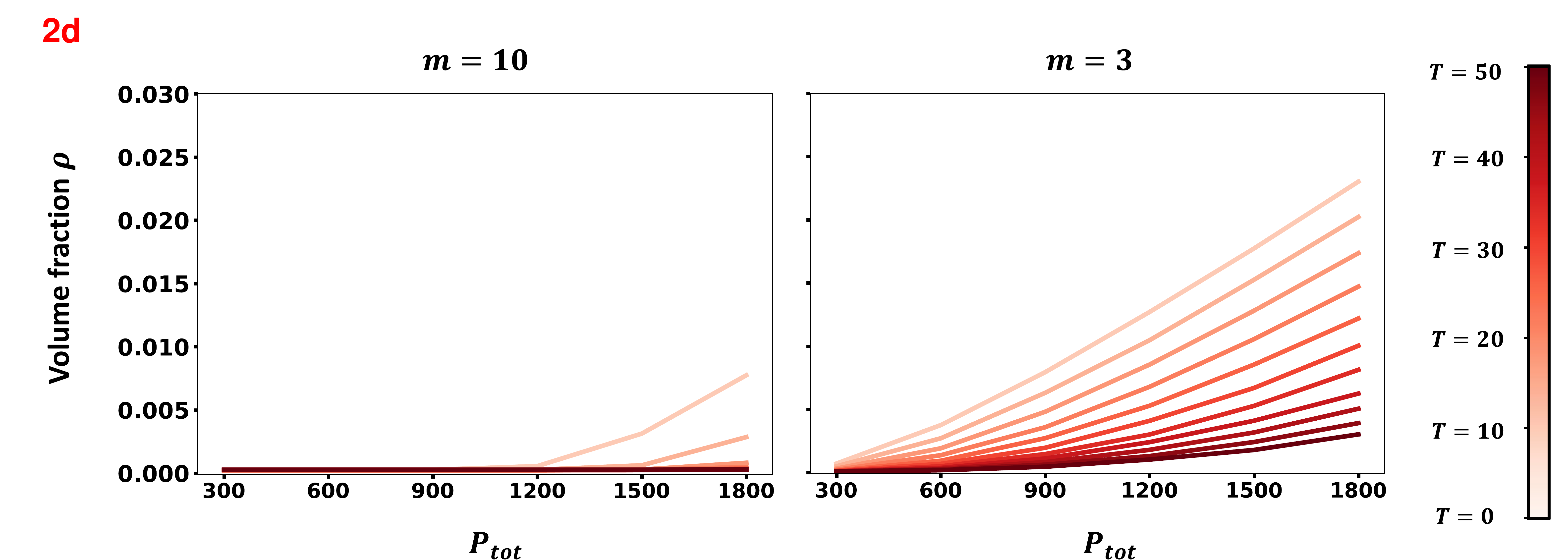}
    \caption{ 
\textbf{Validate the necessity of $m>3$ using the volume fraction.} Volume fractions \eqref{eq:vol ratio} as a function of $P_{tot}$ for $m=10$ and $m=3$. Only the case of $m=10$ displays a plateau on the range of small values $P_{tot}$ that characterizes the threshold protein concentration for droplet formation.}
\label{fig:m=3}
\end{figure}

\subsection{Ostwald ripening}\label{app:ostwald}
Opposed to droplet formation, when proteins leave a droplet, the balance of the fluxes can be collapsed. This leads dilution of the dense state modeled by $D_m\to mP$.  This reaction is also modeling Ostwald ripening, which is another mechanism to grow the size of droplets rather than coalescence and fusion. According to Ostwald ripening, small droplets disappear by the reaction $D_m\to mP$ since higher flux (or higher Laplace pressure) of small droplets. In other words, the molecules on the surface of a smaller droplet are energetically less stable compared to those in larger droplets. Hence the inside proteins leave the droplet \cite{naz2024self}. These proteins diffuse and merge into large droplets as $P+D_k\to D_{k+1}$ for some $k$.

\subsection{Protein assembly and LLPS}\label{app:assembly}
Protein assembly and phase separation are distinct concepts \cite{banani2017biomolecular}. Proteins can not only form droplets but also assemble with other proteins \cite{banani2017biomolecular, bartolucci2024interplay}. Multiple monomers can assemble to create a multimer. 
    A thermodynamical point of view revealed the size distribution of protein assemblies can change the configuration of LLPS \cite{bartolucci2024interplay}. 
    
    Therefore, it is meaningful to model such assembly processes with reactions $P_i+P_j\rightleftharpoons P_{i+j}$, where $P_i$ represents $i$-mers. However, to build a more coarser-grained model, we assume that the assembly process is on a much slower time scale compared to time scale of phase separation. This setting is also used in \cite{bartolucci2024interplay}. Another possible scenario is that the assembly equilibrium is already made, so that the process of the protein assembly is less dynamical than phase separation. Hence we can average out the effect of the size distribution of the assemblies. In those scenarios, we assume that $P$ represents the number of total proteins including both the monomers and multimers. 

    \subsection{Mass-action kinetics under well-mixed compartments}
    Phase separation obviously makes the space demixed. However, each compartment can be well-mixed. That is, the space on a diluted phase is well-mixed and the inner space of droplets are also well-mixed. This condition is essential for the reactions to take place in either diluted spaces or dense spaces. Under this condition of well-mixed compartments, it is reasonable to use mass-action kinetics for the reactions in \eqref{eq:main crn}.

\subsection{Mobilities of droplets and the proteins in the dense phase}\label{sub_app:mobility}
As droplet mass increases, the diffusion coefficient of proteins inside the droplets decreases \cite{kamagata2022guest, snead2022membrane}, which causes small mobility of droplets. Additionally, an existing study provided a more precise comparison between proteins inside and outside droplets \cite{kim2023intrinsically}. The authors  
 experimentally found that the molecular rearrangement rate of membrane-bound proteins is slower within droplets compared to the same proteins outside the droplets \cite{kim2023intrinsically}. This motivated us to assume that droplet fusion and fission events occur at a much slower rate than droplet formation, coarsening, and dissociation events. Thus, we ignore fusion and fission in our model.
Note that we incorporate the disparity of protein mobilities into the reaction rate constants with the constant $\mathcal V$ in the rates $b_k$'s.

\subsection{Multicomponent LLPS}
LLPS often takes place with multiple proteins as scaffold proteins drive phase separation and clients are engaged into it \cite{andre2020liquid}. We simply consider a single type of scaffold proteins in this work for the sake of simplicity.

\subsection{Temperature effects on the volume fraction and the phase diagram}
The temperature effect in the phase diagram can be explained as follows. 
Due to a low viscosity with high temperature, a higher number of proteins is needed to initiate forming droplets and maintain them. However, a longer hydrodynamic radius with high temperature makes the volume of the droplets bigger, so that the ratio of the droplet volume can easily be big with high temperature

\section{Parameters in all figures}\label{app:parameters}

We provide the values of the parameters we used in Table \ref{table}.

The temperature $T$ is measured in Celsius.
The following functions are used to generate the all figures such that 
\begin{align*}
	\mathcal{V} &= \mathcal{V}_0 \cdot e^{\frac{8500}{T+273}}\\
	r &= r_0 \cdot (T+273) 
\end{align*} In general, we use a diffusion coefficient for both 2d and 3d, $\mathcal{D}^{2d} = \mathcal{D}^{3d} = 1$, and set the system size to $R = 1$ for all figures. However, for Figure \ref{fig:free energy and time trajectory}C, we used more realistic parameters $\mathcal{D}^{2d} = 10\mu m^2 s^{-1}$ as found in \cite{snead2022membrane}, a system size of $R = 10^3 \mu m$ and $r = 5\mu m$ at $T = 36$ based on an existing study. For Figure 8, time trajectories are sampled using the same algorithm and initial state described for 2d in Figure \ref{fig:phase diagram} with the sampling process is terminated after $10^4$ reactions for both $m=3$ and $m=10$.

\begin{table*}[h]
    \centering
    \resizebox{\textwidth}{!}{
    \begin{tabular}{|l| l|c|c|c|c|c|c|c|c|}
        \hline
        \makecell{\textbf{Para-} \\ \textbf{meters}} & \textbf{Definition} & \multicolumn{8}{c|}{\textbf{Figures}} \\ \cline{3-10} & & \textbf{\ref{fig:overall}} & \textbf{\ref{fig:free energy and time trajectory}} & \textbf{\ref{fig:phase diagram}} & \textbf{\ref{fig:viscosity}D} & \textbf{\ref{fig:threshold}} & \textbf{\ref{fig:radius}} & \textbf{\ref{fig:m=3}} & \textbf{\ref{fig:gillespie stability}} \\ \hline
	$m$ & Threshold & \multicolumn{4}{c|}{10}& [5,35] & 10 &3, 10 & 10\\ \hline
        $L$ & \makecell[l]{The number of\\ proteins\\ in the largest \\droplets}& \multicolumn{4}{c|}{19} & 36 & 19 & 19 & 19\\ \hline
        $\mathcal{V}_0$ & \makecell[l]{Viscosity \\ scaling constant} & $e^{-25}$ &$e^{-19}$ &$e^{-25}$
          & \multicolumn{2}{c|}{$e^{-6}$}& $e^{-12}$ & \multicolumn{2}{c|}{$e^{-25}$} \\ \hline
        $T$ & Temperature & \multicolumn{2}{c|}{$36$} & $[0,50]$ & $36$  & $10,36,60$ & $36$ & \multicolumn{2}{c|}{$[0,50]$} \\ \hline
        $\frac{r_0}{R}$ & \makecell[l]{Hydrodynamic\\ radius / system size} & \multicolumn{1}{c|}{$\frac{5\cdot 10^{-3}}{309}$} & - & \multicolumn{6}{c|}{$\frac{5\cdot 10^{-3}}{309}$}   \\ \hline
        $P_{tot}$ & Total proteins  & - & $5 \cdot 10^{4}$ & \makecell[l]{   2d) $[5\cdot 10^2,8\cdot 10^3]$ \\ 
          \& $[1.5\cdot10^4,3\cdot10^4]$ \\   3d) $[10^5,2.3\cdot 10^5]$\\
          \& $[45\cdot10^5,85\cdot10^5]$} & $200$ & $100$ & $19$ & \makecell[l]{  2d) \\
        $[3\cdot 10^2,18\cdot 10^2]$}  & \makecell[l]{  2d) $8\cdot 10^3,1.5\cdot 10^4,$ \\
          \& $3\cdot 10^4 $ \\ 
          3d) $23\cdot 10^5,45\cdot 10^5,$\\   \& $65\cdot 10^5$} \\ \hline
 \end{tabular} }
 \caption{Parameter values for all figures}\label{table}
\end{table*}

\section{Derivation of stationary distributions}\label{app:def 0}

In the literature of chemical reaction network theory, people use structural properties of chemical reaction networks to derive dynamical features of the associated dynamical systems for the chemical reaction networks. 
The following theorem (Theorem 4.2 in \cite{AndProdForm}) shows that a certain structural property can imply a closed form of the stationary distribution of the associated Markov chain. 
\begin{thm}\label{thm:def 0}
    Let $X$ be the associated continuous-time Markov chain for a chemical reaction network whose connected components a strongly connected. Let $n$ and $\ell$ denote the numbers of the nodes and the connected components of the chemical reaction network, respectively. Furthermore let $s$ be the dimension of the vector space $span\{y'_k-y_k: y_k\to y_k\}$. For $X(0)=\x_0$, if
    $n-\ell-s=0$,
    then $X$ admits a unique stationary distribution $\pi$ such that for each state $x$,
\begin{align}\label{eq:stationary dist}
       \lim_{t\to \infty}P(X(t)=\x)= \pi(\x)=M \prod_{i=1}^d \frac{c_i^{\x_i}}{\x_i!},
    \end{align}
    where $c$ is any positive steady state of a system of ordinary system given by \begin{align*}
   \frac{d}{dt}x(t)=\sum_{y\to y'}\prod_{i=1}^d(x_i(t))^{y_i} (y'-y),
    \end{align*} where $M$ is the normalizing constant $M=\left( \sum_{\x\in \mathbb S_{\x_0}}\prod_{i=1}^d \frac{c_i^{\x_i}}{\x_i!} \right )^{-1}$ and  $\mathbb S_{\x_0}$ is the closed communication class containing $\x_0$. 
\end{thm}
We need to clarify some terminologies in Theorem \ref{thm:def 0}. We first define connected components as the typical concept in graph theory, regarding the chemical reaction network as a graph. A connected component is strongly connected if there exists a path from node $v$ to node $u$ in the component then there is a path from node $u$ to $v$ in the connected component. For example, in the following reaction network,
\begin{align*}
&A\xrightarrow{ \ \ \ \ \ } B \hspace{.7in} 2C\xrightarrow{ \ \ \ \ \ } D\\[-1ex]
&  \displaystyle \nwarrow \hspace{.2in}  \swarrow  \hspace{2.2cm}  \displaystyle \searrow \hspace{.2in}  \swarrow \\[-1ex]
&\hspace{0.5cm} C \hspace{1.2in} 2B
\end{align*}
the first connected component is strongly connected but the second one is not.

\begin{remark}
   There are $n=2L-2m+2$ nodes in \eqref{eq:main crn} and $\ell=L-m+1$ connected components each of which is strongly connected. Also, the reaction vectors are
    \begin{align*}
        \pm (-m,1,0,\dots,0), \pm (-1,-1,1,\dots,0),\dots, \pm (-1,0,\dots,-1,1),
    \end{align*}
    Hence the dimension of the vector space spanned by these reaction vectors is $s=L-m+1$.  Hence  $n-\ell-s=0$. This implies that the closed form \eqref{eq:pi} of the stationary distribution of the Markov model associated with \eqref{eq:main crn}.
    \end{remark}

\section{Approximate \texorpdfstring{$\pi$} \ \  with sample trajectories}\label{eq:initial_state}
The volume fractions and phase diagrams in Figure \ref{fig:phase diagram} are estimated using the time trajectories sampled with  Gillespie's algorithm (2d) and the tau-leaping method (3d) \cite{cao2005tau_leaping}. In the simulation, the initial state is defined as
 \begin{align*}
    X(0) = \begin{cases}
        (0,P_{tot}/m,0,\dots,0)  \quad \quad &\text{if $d=2$}\\
        (P_{tot}/2,P_{tot}/2m,0,\dots,0)  \quad \quad &\text{if $d=3$}\\
    \end{cases} 
\end{align*}
where the sampling process was terminated after $3 \cdot 10^4$ reactions in 2d simulations. For the 3d case, to reduce computational costs, we used the tau-leaping method and terminated the sampling at $3.5 \cdot 10^4$th reactions.
Figure \ref{fig:gillespie stability} showed that the samples closely approximate the volume fraction \eqref{eq:vol ratio}, as the number of proteins in the dilute phase $P(t)$ stabilizes when the specified number of reactions is fired.

\begin{figure}[ht]
    \centering
\includegraphics[width=1.05\textwidth]{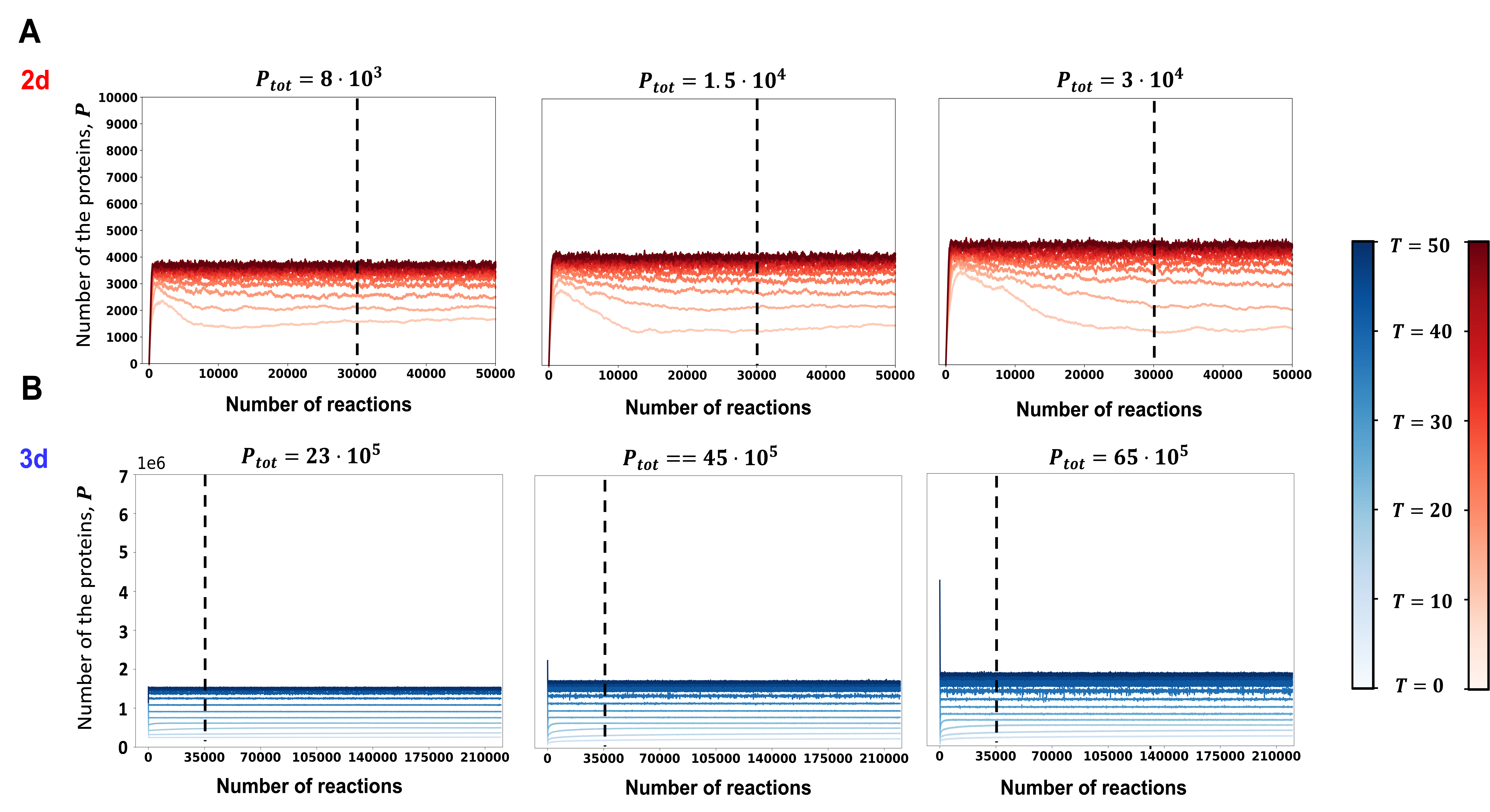}
    \caption{ 
\textbf{Convergence of trajectories.} The time trajectories with the different total numbers of proteins $P_{tot}$ and the different choices of the temperatures for 2d (top) and 3d (bottom). }
\label{fig:gillespie stability}
\end{figure}

\section*{Acknowledgements}
The first and third authors were supported by the National Research Foundation of Korea (NRF) grant funded by the Korea government (MSIT) (No. 2022R1C1C1008491 and No. RS-2023-00219980). The third author was also supported by NRF grant funded by MSIT (No.2021R1A6A1A10042944), POSCO HOLDINGS research fund (2022Q019), and Samsung Electronics Co., Ltd (IO230407-05812-01).
The second author was supported by the National Science Foundation (Grant Nos. DMS-2325258 and CAREER DMS-1944574).

\end{document}